\documentclass[submission,copyright,creativecommons]{eptcs}
\providecommand{\event}{Gandalf 2018} % Name of the event you are submitting to
\usepackage{breakurl}             % Not needed if you use pdflatex only.
\usepackage{underscore}           % Only needed if you use pdflatex.

\usepackage{tikz}
\usetikzlibrary{arrows}
\usepackage{pict2e}
\usepackage{amsmath}
\usepackage{amsthm}
\usepackage{amssymb}
\usepackage{latexsym}
\usepackage{proof}
\usepackage{microtype}
\usepackage{bussproofs}
\usepackage{amssymb}
\usepackage{amsmath}
\usepackage{mathrsfs}
\usepackage{upgreek}
 \usepackage{graphicx}
\usepackage{hyperref}

\newtheorem{theorem}{Theorem}
\newtheorem{lemma}{Lemma}
\newtheorem{proposition}{Proposition}
\theoremstyle{definition}
\newtheorem{definition}{Definition}
\theoremstyle{remark}

\newtheorem{example}{Example}

\newdimen\CdotAxis
\newcommand*{\CdotAux}[3]{%
  {%
    \settoheight\CdotAxis{$#2\vcenter{}$}%
    \sbox0{%
      \raisebox\CdotAxis{%
        \scalebox{#1}{%
          \raisebox{-1.1pt}{%
            $\mathsurround=0pt #2#3$%
          }%
        }%
      }%
    }%
    \dp0=0pt %
    \sbox2{$#2\bullet$}%
    \ifdim\ht2<\ht0 %
      \ht0=\ht2 %
    \fi
    \sbox2{$\mathsurround=0pt #2#3$}%
    \hbox to \wd2{\hss\usebox{0}\hss}%
  }%
}

\newcommand{\proj}                     { {\mathsf{\uppi}} }
\newcommand{\pair}[2]{\langle #1,#2\rangle}

\newcommand{\nf}{\mathsf{NF}}

\newcommand{\Ecrom}[3]{#2 \parallel_{#1} #3}

\newcommand{\sq}[1]  {{\boldsymbol{#1}}}
\newcommand{\nor}[1] {{\mathcal{#1}}}

\newcommand{\comment}[1]{}

\newcommand{\IMPL}{\rightarrow}
\newcommand{\ET}{\wedge}

\newcommand{\FAL}{\bot}

\newcommand{\NON}{\neg}

\newcommand{\impl}{\rightarrow}
\newcommand{\et}{\wedge}
\newcommand{\vel}{\vee}
\newcommand{\fal}{\bot}
\newcommand{\ver}{\top}
\newcommand{\non}{\neg}

\newcommand{\emiddle}{\mathrm{EM}}

\newcommand{\lamg}{\lambda _{\mathrm{G}}}
\newcommand{\lama}{{\lambda _{\mathrm{CL}} }}
\newcommand{\lamem}{{\lambda _{\mathrm{CL}} }}

\newcommand{\inappendix}[1]{}

\newcommand{\lam}{\lambda}
\newcommand{\lan}{\langle}
\newcommand{\ran}{\rangle}
\newcommand{\efq}[2]{#2 \, \mathsf{efq}_{#1 }}
\newcommand{\exfalso}{\mathsf{efq}}
\newcommand{\p}{\parallel}

\newcommand{\send}[1]{\overline{#1}\,}
\newcommand{\inj}{\upiota}
\newcommand{\termt}{\mathrm{t}\hspace{-0.6pt}\mathrm{t}}

\newcommand{\mapstopar}{\rightrightarrows}

\newcommand{\NJ}{\mathrm{NJ}}

\def\NJ{{\bf NJ}}

\def\Deduce#1{\hbox{$\hphantom{#1}$\kern\inferLabelSkip\DeduceSym\kern\inferLab
elSkip$#1$}}

\newcommand\mydots{\!\makebox[1em][c]{.\hfil.\hfil.}}

\newcommand{\noappendix}[1]{\comment{#1}}

\title{Classical Proofs as Parallel Programs\thanks{Supported by FWF: grant Y544-N2 and project W1255-N23.}}
\author{Federico Aschieri
\institute{TU Wien, Vienna, Austria}
\and
Agata Ciabattoni
\institute{TU Wien, Vienna, Austria}
\and
Francesco A. Genco
\institute{TU Wien, Vienna, Austria}
}

\def\titlerunning{Classical Proofs as Parallel Programs}
\def\authorrunning{F. Aschieri, A. Ciabattoni \& F.A. Genco}
\begin{document}
\maketitle

\begin{abstract} We introduce a first proofs-as-parallel-programs
correspondence for classical logic. We define a parallel and more
powerful extension of the simply typed $\lambda$-calculus
corresponding to an analytic natural deduction based on the excluded
middle law.  The resulting functional language features a natural
higher-order communication mechanism between processes, which
also supports broadcasting. The normalization procedure makes use of reductions that implement novel
techniques for handling and transmitting process closures.
\end{abstract}

% The resulting functional language provides a computational
% interpretation to the intrinsic parallel features of the rule for the excluded middle. 

\section{Introduction}

\label{Introduction} The $\lambda$-calculus is the heart of functional
programming languages.  The deep connection between its programs and
intuitionistic proofs is known as Curry--Howard correspondence; useful
consequences of this correspondence are the termination of well-typed
functional programs and the possibility of writing provably correct
programs, see, e.g.,~\cite{Wadler}.

% In this paradigm, computations can be seen as proof transformations,
% like normalization of natural deduction derivations.  At the core of
% functional programming lies the famous correpondence between logic and
% computation.  Known as Curry--Howard correspondence \cite{Howard}, it
% intuitionistic proofs can be seen as simply typed $\lambda$-terms,
% intuitionistic formulas as types, and evaluation of terms can be seen
% as proof manipulations (e.g., normalization of natural deduction
% derivations).  one-to-one correspondence between intuitionistic proofs
% and simply typed $\lambda$-terms, and beetween intuitionistic formulas
% and types. In this paradigm, known as Curry--Howard correspondence
% \cite{Howard}, evaluation of terms can be seen as proof manipulations
% (e.g., normalization of natural deduction derivations).  Useful
% consequences of this correspondence are the termination of

The extension of the Curry--Howard correspondence to classical logic
came many years later, with Griffin's discovery \cite{Griffin} that
Pierce's law $((A \to B) \to A) \to A$ provides a type for the call/cc
operator of Scheme. Since then many $\lambda$-calculi motivated by the
correspondence with classical logic have been introduced. Remarkably,
different formalizations of the same logic lead to different results.
In particular computational interpretations of classical logic are
very sensible to the selected logical formalism and the concrete
coding of classical reasoning. The main two choices for the former are
natural deduction and sequent calculus, while for the latter are
Pierce's Law, reasoning by contradiction $(\lnot A \rightarrow \bot
)\rightarrow A$, multi-conclusion deduction and the excluded middle
law $\emiddle$ $\lnot A \lor A$.

% In particular computational interpretations of classical logic are
% very sensible to: {\it the selected logical formalism}, the main two
% choices being natural deduction and sequent calculus; {\it the
% concrete coding of classical reasoning} in the logical system, some
% natural choices being formalizations using Pierce's Law, reasoning by
% contradiction $(\lnot A \rightarrow \bot )\rightarrow A$,
% multi-conclusion deduction and the excluded middle $\lnot A \lor A$.

For instance, reasoning by contradiction gives rise to control
operators and corresponds very directly to Parigot $\lambda
\mu$-calculus \cite{Parigot}, which relates to classical natural
deduction as the $\lambda$-calculus relates to intuitionistic natural
deduction $\NJ$~\cite{Sorensen}.  Examples of multi-conclusion
deductions in classical sequent calculus are \cite{Herbelin,WadlerSQ}.
These Curry--Howard correspondences match perfectly classical logic
and computation: each step of program reduction corresponds to a proof
transformation, and the evaluation of $\lambda$-terms corresponds to
normalization, a procedure that makes proofs {\em analytic}, i.e.\
only containing formulas that are subformulas of premises and
conclusion.

% Multi-conclusions were first introduced in Gentzen classical sequent
% calculus, which gives rise to a different line of computational
% interpretations of classical logic \cite{Herbelin,WadlerSQ}.  In
% general, control operators model sequential evaluation and extend the
% $\lambda$-calculus via programming concepts such as continuations
% \cite{Krivine}.

%Remarkably, none of them is based on  $\lnot A \lor A$, in spite
%of the implicit \emph{parallelism} contained the excluded middle,
%see~\cite{DanosKrivine}.

%of the implicit \emph{parallelism} contained the excluded middle,
%see~\cite{DanosKrivine}.

% ~\cite{UB2001}

All the functional programming languages resulting from Curry–-Howard
correspondences for propositional classical logic model {\em
sequential computation} and extend the $\lambda$-calculus via
programming concepts such as continuations~\cite{Krivine}.
Remarkably, none of them is based on $\lnot A \lor A$, which appears
to be related to some form of \emph{parallelism},
see~\cite{DanosKrivine}.  A natural way to make this parallelism
explicit would be to extend the Curry--Howard correspondence to $\NJ$
augmented with suitable rules for $\emiddle$.  So far, however, it has
been a long-standing open problem to provide an analytic natural
deduction based on $\emiddle$
\emph{and} enjoying a significant computational interpretation.  In the
only known Curry--Howard correspondence for  $\emiddle$-based
propositional natural deduction~\cite{deGrooteex} proof-terms are not
interpreted as parallel programs and do not correspond to analytic
proofs.  The calculus in~\cite{deGrooteex} extends the
simply typed $\lambda$-calculus with an operator for exception
handling. The lack of analiticity, however, implies that important
reduction rules are missing: there are some exceptions that should be
raised, but are not. These exceptions contain free variables which are
locally bound and cannot be delivered to the exception handler,
otherwise they would become free variables in the whole term.  Hence a
crucial missing element in~\cite{deGrooteex} is \emph{code mobility},
along with the techniques handling the bindings between a piece of
code and its environment~\cite{Fuggetta}. Parallelism and code
mobility were instead employed to define the calculus  $\lamg$~\cite{lics2017}  and provide a
Curry--Howard correspondence for G\"odel logic, a well known logic
intermediate between classical and intuitionistic logic.

%In addition, \cite{deGrooteex} does not exploit the potential parallelism
%in the excluded middle law.
% Hence a crucial missing concept in~\cite{deGrooteex} is \emph{code
% mobility}, a technique that handles precisely the bindings between a piece of
% code and its environment~\cite{Fuggetta}. 

%The result is $\lamem$; the technique that we employ, the Curry--Howard correspondence; 
%the typing system, the excluded-middle-based natural deduction.

%In contrast with $\pi$-calculus, the most widespread formalism for modeling concurrent systems~\cite{Milner, sangiorgiwalker2003},  $\lamem$ is a parallel functional language intended (as a base)
%for programming.%, see Table~\ref{tab:typing}. 

We exploit here the techniques developed for
$\lamg$ in order to extract from propositional
classical logic a new parallel $\lambda$-calculus with a remarkably
simpler communication mechanism. The calculus $\lamem$, as we call it,
extends simply typed $\lambda$-calculus by a communication
mechanism that interprets the natural deduction rule for $A \vee \neg
A$.
Processes of $\lamem$  only communicate through private 
channels  similar  to those bound by the  restriction 
operator $\nu$ in the $\pi$-calculus,
% The communication channels of $\lamem$  are private 
% and behave similarly to those generated by the restriction 
% operator $\nu$ in the $\pi$-calculus,
the most widespread
formalism for modeling concurrent systems~\cite{Milner,
sangiorgiwalker2003}.  These channels are introduced by the
typing rule for $\emiddle$ and their behavior during
communication is defined by  $\lamem$ reduction rules. The
basic communication reductions, called {\em basic cross reductions}, behave as follows:
let $ \mathcal{C}$ be a process ready to send a
message to another process $\mathcal{D}$ through a channel $a$, in
symbols $\; \mathcal{C}[\send{a}\, u]\parallel_{a} \mathcal{D}\; $;
then, if
$u$ is  data or a closed process, the result of the communication is
$\mathcal{D}[ u /a ] $, similarly to asynchronous
$\pi$-calculus~\cite{HT} or the concurrent $\lam$-calculus proposed in~\cite{Boudol89}.  Although simple, this mechanism makes $\lamem$ more powerful than simply typed
$\lambda$-calculus and propositional $\lam _\mu$
%{\ehi and de Groote's $\lam_{exn}$ (comparison removed because it is difficult to study formally)}
(see Prop.~\ref{power}).  Remarkably it can also model
\emph{races}, situations in which several processes compete for a
limited amount of resources;
%, as in client/server interactions; 
moreover, in contrast with calculi based on point-to-point
communication \cite{EM} as $\pi$-calculus or $\lamg$~\cite{lics2017} (see 
Sec.~\ref{Sec:exprpower}) $\lamem$ also renders \emph{broadcasting}.

The normalization of $\lamem$ ensures that natural deduction
derivations can be transformed into analytic ones. Similarly to
the normalization for $\lamg$, this procedure requires additional
reduction rules ({\em cross reductions}) that offer a solution to a
``{\em fundamental problem in any distributed implementation of a
statically-typed, higher-order programming language}'': how to send
open processes and transmit function closures from one node to another
(cf.\ the description of Cloud Haskell in \cite{EBPJ2011}).  Indeed,
the process $u$ to be sent might not be closed and might need
resources that are or will be available in $ \mathcal{C}$.  {\em Cross
reductions} allow $u$ to be transmitted and create a new communication
channel $b$ for transferring the complete closure of $u$ afterwards.
Not even higher-order $\pi$-calculus directly supports such a
mechanism; let alone $\pi$-calculus, whose communication only applies
to data or channels.  As discussed in Ex.~\ref{ex:optimize}, our code
mobility can be used for program optimization.  The technically
sophisticated normalization proof in Sec.~\ref{sec:norm} adapts the
proof for $\lamg$~\cite{lics2017}.

\section{$\lamem$: a Curry--Howard interpretation of Classical Logic}
\label{sec:lamemcalc}

We introduce $\lamem$, our typed parallel $\lambda$-calculus for
Classical Logic. $\lamem$ extends the standard Curry--Howard
correspondence~\cite{Howard,Wadler} for the intuitionistic natural
deduction $\NJ$~\cite{Prawitz} by a parallel operator that interprets the rule
$(\emiddle)$ for $A \vee \neg A$. 
% (see Table~\ref{tab:typing})
%Then we describe $\lamem$-terms and their computational behavior.
Table~\ref{tab:typing} defines a type assignment for $\lamem$-terms,
called \textbf{proof terms} and denoted by $t, u, v \dots$, which is
isomorphic to the natural deduction system $\NJ + (\emiddle)$.  The typing
rules for axioms, implication, conjunction, disjunction and
ex-falso-quodlibet are those for the simply typed
$\lambda$-calculus~\cite{Girard}. Parallelism is introduced by the
$(\emiddle )$ rule .  
The contraction rule, useful to define parallel terms that do not communicate,
is in fact redundant from the  point of view of proof theory.
%The contraction rule, the second of Table~\ref{tab:typing},

% , and is the analogous of the external contraction rule $(ec)$ of
% hypersequent calculi, see Sec.~\ref{sec:nd}, but it can be omitted.

\begin{table}[h!]
\hrule 
\smallskip

  \centering
  \begin{footnotesize}
  \begin{center}
    $\begin{array}{c} x^A: A \\ %\text{for $x$ intuitionistic variable}
     \end{array}$ 
$ \qquad \vcenter{\infer[(\text{contr.})]{ u \p v : A}{ u : A
& v : A }}$ $\qquad \vcenter{\infer{ \langle u,t\rangle: A \wedge B}{u:A & t:B}}
     \qquad \vcenter{\infer{u\,\pi_0: A}{u: A\wedge B}} \qquad
     \vcenter{\infer{u\,\pi_1: B}{u: A\wedge B}}$
\bigskip

$ \vcenter{\infer{\lambda x^{A} u: A\rightarrow
B}{\infer*{u:B}{[x^{A}: A]}}} \quad\qquad \vcenter{\infer{tu:B}{ t: A\IMPL
B & u:A} } \hspace{200pt}\begin{array}[c]{c} \infer{
\efq{P}{u}: P}{ u: \bot} \\ \text{ with } P \text{ atomic and } P \neq
\bot
\end{array}$

\vspace{-50pt}

% {\ehhi      $\infer{\inj_{0}(u): A\vee B }{u: A} \quad $
%      $\infer{\inj_{1}(u): A\vee B}{u: B}\quad $
%      $\infer{u\, [x^{A}.w_{1}, y^{B}.w_{2}]: C}{ u: A\lor B &
%        \infer*{w_{1}: C}{[x^{A}: A]} & \infer*{w_{2}: C}{[y^{B}: B]}}$
% }    
     \smallskip
      
%$ a^A :A
%\qquad $
%$\vcenter{\infer{a^{\non A} t :
%    \fal}{t:A}}\qquad 
%    $
    $\hspace{50pt} \vcenter{\infer[(\emiddle )]{u \parallel_{a} v:
B}{\infer*{u:B}{[a^{\NON A}: \non
A]}&\infer*{v:B}{[a^{ A}: A]}}}  $
\end{center} \vspace{-8pt}where all the occurrences of $a$ in $u$ and $v$ are respectively of the form $a^{\non A}$ and $a^{A}$
   \end{footnotesize}
\smallskip
\hrule 
\medskip
\caption{Type assignments for $\lamem$ terms and natural deduction rules.}\label{tab:typeem}
\label{tab:typing}
\end{table}

The reduction rules of $\lamem$ are in Figure~\ref{fig:redem}. They
consist of the usual simply typed $\lambda$-calculus reductions,
instances of $\vel$ permutations adapted to the $\p$ operator, and new
communication reductions: {\it Basic Cross Reductions} and {\it Cross Reductions}. Since we are dealing with a Curry--Howard
correspondence, every reduction rule of $\lamem$ corresponds to a
reduction for the natural deduction calculus $\NJ + (\emiddle)$. % (see Table~\ref{tab:typing}).

Before explaining the calculus reductions, we recall the essential
terminology and definitions, see e.g.~\cite{Girard}. Proof terms may contain variables
$x_{0}^{A}, x_{1}^{A}, x_{2}^{A}, \ldots$ of type $A$ for every
formula $A$; these variables are denoted by $x^{A},$ $ y^{A}, $ $
z^{A} , \ldots,$ $ a^{A}, b^{A}, c^{A}$ and whenever the type is
irrelevant by $x, y, z, \ldots, a, b$. For clarity, the $(\emiddle )$ rule
variables will be often denoted by $a, b, c, \dots$ but they are not in a syntactic
category apart.  A variable $x^{A}$ that occurs in a term of the form
$\lambda x^{A} u$ is called \textbf{$\lambda$-variable} and a variable
$a$ that occurs in a term $u\parallel_{a} v$ is called
\textbf{channel} or \textbf{communication variable} and represents a
\emph{private} communication channel between the processes
$u$ and $v$. We adopt the convention that $a^{\non A}$ and $
a^{A}$ are denoted by $\send{a}$ and $a$ respectively, where
unambiguous.
Free and bound variables of a term are defined as usual. For the new
term $\Ecrom{a}{u}{v}$, all free occurrences of $a$ in $u$ and $v$ are
bound in $\Ecrom{a}{u}{v}$. We assume the standard renaming rules and
$\alpha$-equivalences that are used to avoid the capture of variables
in the reductions.

\noindent \textbf{Notation}. The connective $\rightarrow$ and $\et $ associate
to the right and by $\langle t_{1}, t_{2}, \ldots, t_{n}\rangle $ we
denote the term $\langle t_{1}, \langle t_{2}, \ldots \langle t_{n-1},
t_{n}\rangle\ldots \rangle\rangle$ (which is $\termt : \ver $ if
$n=0$) and by $\proj_{i}$, for $i=0, \ldots, n$, the sequence
$\pi_{1}\ldots \pi_{1} \pi_{0}$ selecting the $(i+1)$th element of the
sequence.  The expression $A_{1} \ET \dots \ET A_{n}$ denotes $\top$
if $n=0$.
 % and thus it is the empty conjunction
As usual, we use $\NON
A $ as shorthand notation for $A \IMPL \FAL$.

We write $\Gamma\vdash t: A$ if $\Gamma= x_{1}: A_{1}, \ldots, x_{n}:
A_{n}$ and all free variables of a proof term $t: A$ are in $x_{1},
\ldots, x_{n}$. From the logical point of view, $t$ represents a
natural deduction of $A$ from the hypotheses $A_{1}, \ldots, A_{n}$.
If the symbol $\parallel$ does not occur in $t$, then
$t$ is a \textbf{simply typed $\lambda$-term} representing an
intuitionistic deduction.

% We define as usual the notion of context $\mathcal{C}[\ ]$ as the part
% of a proof term that surrounds a hole, represented by some fixed
% variable. In the expression $\mathcal{C}[u]$ we denote a particular
% occurrence of a subterm $u$ in the whole term $\mathcal{C}[u]$. We
% shall just need those particularly simple contexts which happen to be
% simply typed $\lambda$-terms.

We first explain the cross reductions from a proof-theoretical point
of view.  Basic cross reductions correspond to the following 
transformation of natural deduction derivations
%reduction of natural deduction derivations
\begin{small}
  \[\vcenter{\infer{C}{\infer*{C}{\infer{\FAL}{[\NON A] &
\deduce{A}{\delta}}} && \infer*{C}{[A]}} }\quad \mapsto \quad
\vcenter{ \infer*{C}{\deduce{A}{\delta}}} \]
\end{small}where no assumption in $\delta$ is discharged below $\bot$
and above $C$. When this is the case, intuitively, the displayed
instance of $(\emiddle )$ is hiding some redex that should be reduced. The
reduction precisely exposes this potential redex~\cite{Prawitz} and we
are thus able to reduce it. More instances of $\NON A$ and $A$ might
occur in the respective branches. This, in combination with the
contraction rule (see Table~\ref{tab:typing}), gives rise to races and
broadcasting, as explained in the computational interpretation
below. Cross reductions correspond to
\begin{footnotesize}
  \[ \vcenter{\infer{C}{\infer*{C}{\infer{ \fal}{[\NON A]^1 &
            \deduce{A}{\deduce{\delta}{[\Gamma]}}}}&&\infer*{C}{[A]^1}}
    }\qquad \mapsto \qquad \vcenter{\infer[{}^2]{ C}{
        \infer[{}^1]{C}{\infer*{ C}{\infer{ \FAL}{[\NON\bigwedge
              \Gamma]^2 & \infer={\bigwedge \Gamma }{ [\Gamma]}
            }}&&\infer*{C}{[A]^1}} &&\infer*{C}{\deduce{A}{\deduce
            {\delta}{\infer=[]{\Gamma}{[ \bigwedge \Gamma]^{2}
              }}}}}}\]
\end{footnotesize}As before, in the right derivation we prove by
$\delta$ all assumptions $A$, but now we also need to
discharge the assumptions $\Gamma$ of $\delta$ in the rightmost
branch; which are discharged in the left derivation between $\bot$ and
$C$. This is done by  $2$: a new application of $( \emiddle )$ to
the conjunction  $\bigwedge\Gamma $ of such assumptions. Accordingly,
we use {\small $ \AxiomC{$ \lnot \bigwedge\Gamma$}
\AxiomC{$\bigwedge\Gamma$} \BinaryInfC{$\bot$} \DisplayProof $ } in the
leftmost branch.  To discharge the remaining occurrences of
$\lnot A$, we need to keep the original instance $1$ of $( \emiddle ) $,
and thus the central branch of the resulting proof is just a
duplicate.

Before discussing the computational content of the calculus we
introduce a few more definitions. 

\begin{definition}[Simple Parallel Term]\label{def:simpparterm} A
\textbf{simple parallel term} is a $\lamem$-term $t_{1}\p\dots \p t_{n}$, where each
$t_{i}$, for $1\leq i\leq n$, is a simply typed $\lambda$-term.
\end{definition}

\begin{definition}\label{defisimplec}\label{def:simpleparallelcont}
A \textbf{context} $\mathcal{C}[\ ]$ is a
$\lamem$-term with some fixed variable $[\;]$ occurring exactly
once. 
\begin{itemize}
\item A \textbf{simple context} is a context which is a simply typed
$\lambda$-term.
\item A \textbf{simple parallel context} is a context which is a simple
parallel term.
\end{itemize}
For any $\lamem$-term $u$ with the same type as $[\;]$,
$\mathcal{C}[u]$ denotes the term obtained replacing $[\;]$ with $u$ in
$\mathcal{C}[\ ]$, \emph{without renaming bound variables}.
\end{definition}

 \begin{definition}[Multiple Substitution]\label{defimultsubst}
Let $u$ be a proof term, $\sq{x}=x_{0}^{A_{0}}, \ldots, x_{n}^{A_{n}}$ a sequence of  variables and $v: A_{0}\land \ldots \land A_{n}$. The substitution 
$u^{v/ \sq x}:=u[v\,\proj_{0}/x_{0}^{A_{0}} \ldots \,v\,\proj_{n}/x_{n}^{A_{n}}]$
replaces each variable $x_{i}^{A_{i}}$ of
any term $u$ with the $i$th projection of $v$.
\end{definition}

% \noindent \textbf{Basic Cross reductions}   $\qquad\qquad\qquad
% \mathcal{C}[\send{a}\, u]\parallel_{a} \mathcal{D} \ \mapsto \
% \mathcal{D}[ u /a ] $  \\ can be fired whenever the free variables of  $u$ are also free in $\mathcal{C}[\send{a}\, u]$. In particular,  $u$  represents executable code or data
% and  directly replaces all occurrences of the  channel endpoint $a$,  as in the reduction
% axiom of the asynchronous $\pi$-calculus~\cite{HT}.

 \paragraph{Basic Cross Reductions} can be fired
whenever the free variables of $t$ are also free in
$\mathcal{C}[\send{a}\, t]$. 
In particular, $t$ may represent executable
code or data and directly replaces all occurrences of the channel
endpoint $a$. In case there is only one sender and one receiver, the
reduction \[\mathcal{C}[\send{a}\, t]\parallel_{a} \mathcal{D}
    \ \mapsto \ \mathcal{D}[ t /a ] \]corresponds to the reduction axiom of the
asynchronous $\pi$-calculus~\cite{HT}.
 In general, $\mathcal{C}[\send{a}\, t]$ has the shape \[{\mathcal C}_1 \p \dots \p {\mathcal C}_i[\send{a} \, t] \p \dots
     \p {\mathcal C}_n \]where more than one process might have a message to send. In this
case, there is a \emph{race} among the processes $\mathcal{C}_{j}$
that contain some message $\send{a}t_{j}$ and compete to transmit it
to $\mathcal{D}$.
% a simple broadcast of a closed program $t$ from a process
% non-deterministically selected among the competing senders (processes
% containing $a:\NON A$)
The sender ${\mathcal C}_i[\send{a} \, t]$ is selected
non-deterministically and communicates its message to $\mathcal{D}$:
\begin{small}
  \[ ({\mathcal C}_1 \p \dots \p {\mathcal C}_{i}[\send{a} \, t] \p
    \dots \p {\mathcal C}_n) \p_{a} {\mathcal D} \; \mapsto \;
    \mathcal{D}[t/a]\]
\end{small}Since the receiving term ${\mathcal D}$ exhausts all its channels
$a:A$ to receive $t$, we remove all processes containing $\send{a} :
\NON A$ and obtain a term without % instances of the channel
$a$.
%This reduction is similar to the communication of asynchronous $\pi$-calculus~\cite{HT}. 
We also point out that $\mathcal{D}$ is an arbitrary term, so it may well be a sequence of parallel process $\mathcal{D}_{1} \parallel \dots \parallel \mathcal{D}_{m}$. In this case, $\mathcal{C}_{i} [\send{a} \, t]$ \emph{broadcasts} its message $t$ to  $\mathcal{D}_{1}, \dots, \mathcal{D}_{m}$:
\[\mathcal{D}[t/a]=\mathcal{D}_{1}[t/a] \parallel \dots \parallel \mathcal{D}_{m} [t/a]\]

% Subject reduction will
% nevertheless be guaranteed.  

\paragraph{Cross Reductions}
%As discussed in the presentation of Cloud Haskell~\cite{EBPJ2011},
%\emph{the question of how to transmit function closures from one node to another is a fundamental one; it must be addressed by any
%distributed implementation of a statically-typed, higher-order
%programming language}.
address a crucial problem of functional languages
with higher-order message exchange: transmitting function closures  (see \cite{EBPJ2011}). 
%open processes
%Logical types, in our framework, automatically provide a solution. 
The solution provided by our logical types is that 
function closures are transmitted
in two steps: first, the function code, 
% code of the function,
then, when it is available,  the evaluation environment.  
%As a result, in $\lamem$ we can communicate open $\lambda$-terms 
As a result, we can communicate open $\lamem$-terms 
which are closed in their original environment, and fill later their free variables, when they will be instantiated. The process of handling and transmitting function closures is typed by  a new instance of $(\emiddle)$.
%Such communication of program closures is discussed in~\cite{EBPJ2011} as a central issue for the development of a
%concurrent version of Haskell. %In $\lamem$, subject reduction is
%guaranteed also when we need to handle the closure of a message. 
 For example: 
%as messages open $\lambda$-terms which are closed in their original
%environment, and to fill their free variables later, when they will be
%instantiated. Such communication of program closures is discussed
%in~\cite{EBPJ2011} as a central issue for the development of a
%concurrent version of Haskell. In $\lamem$, subject reduction is guaranteed also when we need to
%handle the closure of a message. } To see how this is possible, 
assume that a channel $a$ is used to send an arbitrary sub-process
$u$ from a 
process ${\mathcal C}$ to a process ${\mathcal D}$ (below left). Since
$u$ might not be closed, it might depend
on its environment for providing values for its free variables --
$\sq{y}$ in the example is bound by a $\lambda$ outside $u$. This
issue is solved in the cross reduction by a fresh channel $b$ which
redirects $\sq{y}$ -- the remaining part of the closure -- to the new
location of $u$ (below right).
\begin{center}
\includegraphics[width=0.22\textwidth]{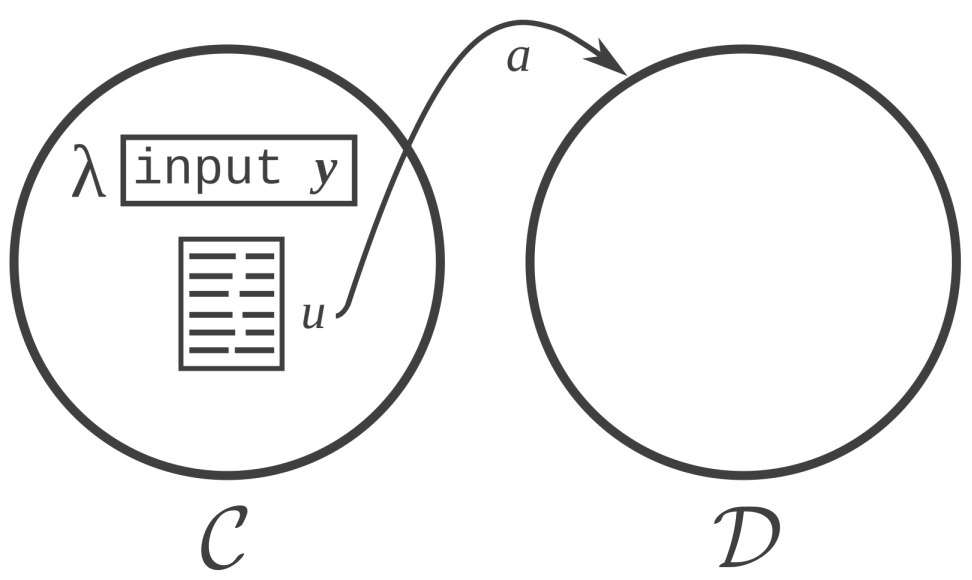}
\hspace{50pt}
\includegraphics[width=0.35\textwidth]{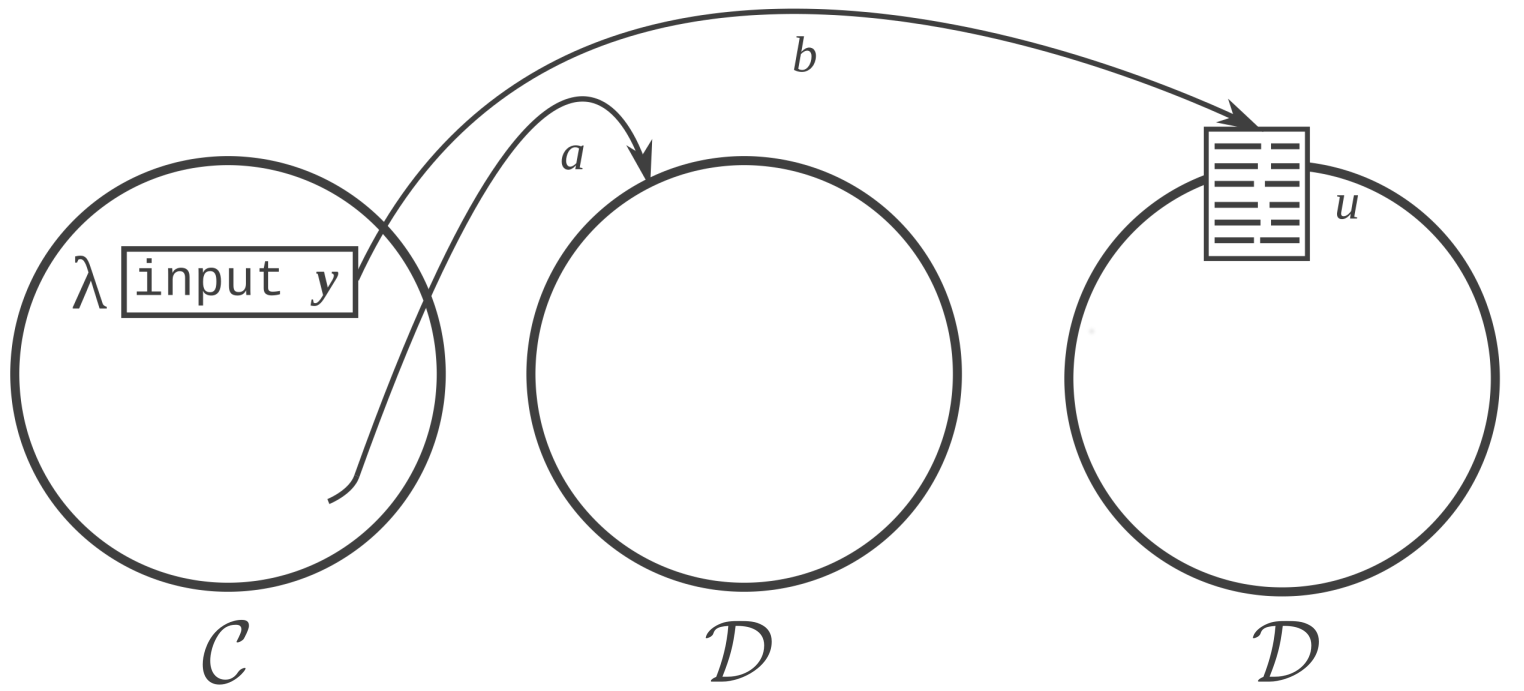}
\end{center} The old channel $a$ is kept for further
messages that $\mathcal{C}$ might want to exchange with $\mathcal{D}$. 
% \[ \text{ Technically, the cross reduction has this shape: }
%   \quad \begin{footnotesize} \mathcal{C}[\send{a}\, u]\parallel_{a}
%     \mathcal{D} \ \mapsto \ ( \mathcal{C}[\send{b} \, \lan \sq{y} \ran
%     ] \parallel_{a} \mathcal{D} ) \, \parallel_{b}\, \mathcal{D}[ u^{b
%       / \sq{y} } /a ]
% \end{footnotesize}\] 
Technically, the cross reduction has this shape: $\quad ( \dots \p {\mathcal C}[\send{a} \, u] \p \dots ) \p_{a}
    \mathcal{D} \mapsto \ ( \mathcal{C}[\send{b} \, \lan \sq{y} \ran
    ] \parallel_{a} \mathcal{D} ) \, \parallel_{b}\, \mathcal{D}[ u^{b
      / \sq{y} } /a ]\quad $
% In order to handle the closure of the migrating
% process, such reduction needs to keep one copy of the original
% communication system $a$ containing the sender ${\mathcal C}$, on
% which $t$ still depends.
On the one hand, the open term $u$ is replaced in $\mathcal{C}$ by the new
channel $\send{b}$ applied to the sequence $\sq{y}$ of the free
variables of $u$; on the other hand, $u$ is sent to the term
 $\mathcal{D}$ as $u^{b /
 \sq{y}}$, so its free variables are removed and replaced by the
channel endpoint $b$ that will receive their future instantiation.

% {\ehhi

% \noindent \textbf{Cross Reductions} $\; $  
% % While basic cross
% % reductions enable non-deterministic message passing, as in $\lamem$,
% % cross reductions in $\lama$ 
% implement a communication 
% mechanism that handles the closures of migrating open
% processes whose free variables are bound in the environment: 

% }

%    $(u\parallel_{a} v)\parallel_{b} w \mapsto (u\parallel_{b}
%w)\parallel_{a} (v\parallel_{b} w) $ and $w \parallel_{b}
%(u\parallel_{a} v) \mapsto (w\parallel_{b} u)\parallel_{a}
%(w\parallel_{b} v)$. 

\paragraph{Communication Permutations}  The only permutations for $\p$
that are not standard $\vel$-permutation-like are \[(u\parallel_{a} v)\parallel_{b} w \mapsto (u\parallel_{b}
w)\parallel_{a} (v\parallel_{b} w) \quad \text{ and } \quad w \parallel_{b}
(u\parallel_{a} v) \mapsto (w\parallel_{b} u)\parallel_{a}
(w\parallel_{b} v)\]These kind of permutations are between parallel operators themselves
and address the \emph{scope extrusion} issue of private channels. For instance, let us consider the term
$(v\parallel_{a} \mathcal{C}[b\, a] )\parallel_{b} w $
Here the process $\mathcal{C}[b\, a]$ wishes to send the channel $a$ to $w$ along the channel $b$, but this is not possible being the channel $a$ private. This issue is solved in the $\pi$-calculus using the congruence $\nu a (P\, |\, Q)\, |\, R \equiv \nu a (P\, |\, Q\, |\, R) $.
 %provided that  $a$ does not occur in $R$, condition that can always be satisfied  by $\alpha$-conversion. 
 Classical logic offers and actually forces a different solution, which is not just permuting $w$ inward but also duplicating it:
 \[
\begin{small}
  ( v \parallel_{a} \mathcal{C}[b\, a])\parallel_{b} w \mapsto
     (v \parallel_{b} w)\parallel_{a} (\mathcal{C}[b\, a]\parallel_{b}
     w) 
 \end{small} 
 \quad \text{ after this reduction $\mathcal{C}[b\, a]$ can send $a$
   to $w$. } \]

We provide now the last definitions needed to formally define the reduction rules of $\lamem$. We start with the notion of \emph{strong subformula}, which is
key for proving Normalization (Section~\ref{sec:norm}). 
%The technical motivations will become clear in Section~\ref{sec:norm}, but 
%Intuitively the new types created by cross reductions must be always strong subformulas of
%already existing types.  
%To define the concept of strong subformula we also need the following definition.
\begin{definition}[Prime Formulas and Factors \cite{Krivine1}]
  A formula is said to be \textbf{prime} if it is not a
  conjunction. Every formula is 
  a conjunction of prime
  formulas, called \textbf{prime factors}. 
\end{definition}

\begin{definition}[Strong Subformula~\cite{lics2017}]\label{definitionstrongsubf}
$B$ is said to be a \textbf{strong subformula} of a formula $A$, if $B$ is a proper subformula of some prime proper subformula  of  $A$. 
\end{definition}

Note that here prime formulas are either atomic formulas or arrow
formulas, so a strong subformula of $A$ must be actually a proper
subformula of an arrow proper subformula of $A$.
The following characterization from~\cite{lics2017} of the strong subformula relation will be often used.  
\begin{proposition}[Characterization of Strong Subformulas]\label{propositionstrongsubf}
If $B$ is a strong subformula of $A$:
\begin{itemize}
\item if $A=A_{1}\land \mydots \land A_{n}$, $n>0$ and $A_{1}, \mydots,
  A_{n}$ are prime, then $B$ is a proper subformula of some $A_{1}, \mydots, A_{n}$; 
\item if $A=C\rightarrow D$, then $B$ is a proper subformula of a prime factor of $C$ or $D$. 
\end{itemize}
\end{proposition}

\begin{figure*}[h]
\vspace{-5pt}
\hrule
\begin{footnotesize}

\smallskip

  \begin{center}

    \smallskip

% {\ehhi   \begin{flushleft} 
%  \textbf{***REMOVE}  \textbf{Activation Reduction} 
%     $ \qquad \pp{a}{ u_{1}\p \dots \p u_m} \; \mapsto \; \pp{a}{ u_1 [b /
%       a]\p \dots \p u_{m} [b/a]}$

%     where $a$ is not active, $b$ is a fresh \emph{active} variable,
%     and there is some occurrence of $a$ in $u$ or in $v$ of the form
%     $a w$, for a value $w$.
%   \end{flushleft}
% }

 % $\pp{a} { {\mathcal C}_1 \p \dots \p {\mathcal C}_i[\send{a} \, t] \p \dots \p {\mathcal C}_m} \; \mapsto \; \pp{a} {{\mathcal C}_1 \p
    %   \dots \p {\mathcal C}_i[a^{\NON A } \, t] \p \dots \p
    %   {\mathcal C}_j[t / a] \p \dots \p {\mathcal C}_m [t / a] }$

  \textbf{Intuitionistic Reductions}

  $(\lambda x^{\scriptscriptstyle A}\, u)t\mapsto
  u[t/x^{\scriptscriptstyle A}] \qquad \pair{u_0}{u_1}\,\pi_{i}\mapsto
  u_i, \mbox{ for $i=0,1$} $
 % {\ehhi  $\qquad  \inj _i (t) [x_0.u_0,x_1.u_1] \mapsto u_i [t/x_i]$}

  \smallskip
 
% {\ehhi 
%   \textbf{Disjunction Permutations}

%   $t [x_0.u_0,x_1.u_1] \xi \mapsto t [x_0.u_0\xi,x_1.u_1\xi] \qquad
%   \text{ if }\xi\text{ is a one-element stack}$
% }

\smallskip

  \textbf{Parallel Operator Permutations}

  $w(\Ecrom{a}{u}{v}) \mapsto \Ecrom{a}{wu}{wv}  $ if $a$ does not occur free in
    $w  $

$(\Ecrom{a}{u}{v}) \xi \mapsto \Ecrom{a}{u \xi}{v\xi} $ if $\xi$ is a
one-element stack and $a$ does not occur free in $\xi $

$w(\Ecrom{}{u}{v}) \mapsto \Ecrom{}{wu}{wv}  $   $\quad (\Ecrom{}{u}{v}) \xi \mapsto \Ecrom{}{u \xi}{v\xi} $ 
$ \quad \lambda x^{\scriptscriptstyle A}\,(\Ecrom{a}{u}{v}) \mapsto
 \Ecrom{a}{\lambda x^{\scriptscriptstyle A}\,u}{\lambda
    x^{\scriptscriptstyle A}\, v} \quad $ 
 % {\ehhi $ \inj _i ( u \p _a v ) \mapsto \inj _i (u) \p _a \inj _i(v)\quad $}
$ \langle u \parallel_a v,\, w\rangle \mapsto \langle u,
w\rangle \parallel_a \langle v, w\rangle \quad $ $\langle w, \,
u \parallel_a v\rangle \mapsto \langle w, u\rangle \parallel_a
\langle w, v\rangle$
$ \lambda x^{\scriptscriptstyle A}\,(u \p v) \mapsto  \lambda
x^{\scriptscriptstyle A}\,u \p \lambda x^{\scriptscriptstyle A}\, v $
$\quad \langle u \parallel v,\, w\rangle \mapsto \langle u,
w\rangle \parallel \langle v, w\rangle \quad $ $\langle w, \,
u \parallel v\rangle \mapsto \langle w, u\rangle \parallel
\langle w, v\rangle$

  $(u\parallel_{a} v)\parallel_{b} w \mapsto (u\parallel_{b}
w)\parallel_{a} (v\parallel_{b} w)$ if the communication complexity of
$b$ is greater than $0$

  $w \parallel_{b} (u\parallel_{a} v) \mapsto (w\parallel_{b}
u)\parallel_{a} (w\parallel_{b} v) $ if the communication complexity
of $b$ is greater than $0$

 $(u\parallel v)\parallel_{b} w \mapsto (u\parallel_{b}
w)\parallel (v\parallel_{b} w)$ if the communication complexity of
$b$ is greater than $0$

  $w \parallel_{b} (u\parallel v) \mapsto (w\parallel_{b}
u)\parallel (w\parallel_{b} v) $ if the communication complexity
of $b$ is greater than $0$

  \smallskip

  \textbf{Communication Reductions}
\end{center}
% {\ehhi \begin{flushleft} \textbf{***REMOVE} \textbf{Activation
% Reduction} $ \qquad\qquad \qquad u \parallel_{a} v \; \mapsto \; u
% [b/a] \parallel_{b} v [b/a]$ where $a$ is not active, $b$ is a fresh
% \emph{active} variable, and there is some occurrence of $a$ in $u$ of
% the form $a w$ for a value $w$.
% \end{flushleft} }
\textbf{Basic Cross Reductions}
$\qquad\qquad\qquad\qquad\qquad  \mathcal{C}[\send{a}\, u]\parallel_{a} \mathcal{D} \ \mapsto \
\mathcal{D}[ u /a ] $

 where $\send{a} : \non A, a : A $, $\mathcal{C}[\;]$ is a
  context; %{\ehi ($\mathcal{C}[\;]$ puo essere arbitrario visto che la condizione seguente risolve la scope extrusion)};
  the free variables of $u$ are also free
in $\mathcal{C}[\send{a}\, u]$; %$\send{a}$ is rightmost in $ {\mathcal C}[\send{a} u]$. 
$\send{a}$ does not occur in $u$ ; and the communication complexity of
$a$ is greater than $0$.

% ; and the communication
% complexity of $a$ is greater than $0$

  \textbf{Cross Reductions}  
    $\qquad\qquad u\parallel_{a}v \mapsto u$ if $a$ does not occur in $u$ and
    $u\parallel_{a}v \mapsto v$ if $a$ does not occur in $v$
    \[(
 \dots \p {\mathcal C}[\send{a} \, u] \p \dots ) \p_{a} \mathcal{D} \mapsto \ (
\mathcal{C}[\send{b} \, \lan \sq{y} \ran ] \parallel_{a} \mathcal{D} )
\, \parallel_{b}\, \mathcal{D}[ u^{b / \sq{y} } /a ]\]

where  $\send{a} : \non A, a : A$; $(
 \dots \p {\mathcal C}[\send{a} \, u] \p \dots )$  is a
 normal simple parallel term;
%  $\mathcal{C}[a\, u],
% \mathcal{D}$ are normal simply typed $\lambda$-terms, $\mathcal{C}$
% a simple context;
$ \sq{y}$ is the  non-empty sequence of free
variables of $u$ bound in $\mathcal{C}[\send{a}\,u]$; $B$ is the
conjunction of the types of the variables in $\sq{y}$ and $b / \sq{y}
$ is a multiple substitution of these variables;  $a$ is
rightmost in ${\mathcal C}[\send{a} \, u]$; $b$ is fresh; $\send{b} : \non B, b : B$; the
communication complexity of $a$ is greater than $0$.  

\smallskip

\end{footnotesize}
\hrule
\caption{Reduction Rules for $\lama$}\label{fig:red}\label{fig:redem}
\vspace{-10pt}
\end{figure*}

Unrestricted cross reductions do not always
terminate. 
%Consider, for example, the loop
%Unrestricted cross reductions do not always terminate. 
Consider, for example, the following loop \begin{equation} \begin{footnotesize}
\label{examplecomplexity} \lambda y^{B}\, a^{\non B}\, y \p_{a}
x^{B \impl \non B} \, a^{B} \;\mapsto\; ( \lambda y\, b\, y \p_{a} x\,
a) \p _b x\, b \; \mapsto\; \lambda y\, b\, y \p _b x\,
b \end{footnotesize} \end{equation}
To avoid such situations we need conditions on the application 
of cross reductions. As shown below, our conditions are based on
%We now need to determine
the complexity of the  channel $a$ of a term $u\parallel_{a} v$, and are 
determined using logic.  We consider the type $B$ such that
$a$ occurs with type $\NON B$ in $u$ and thus with type $B$ in $v$,
the type $A$ of the term $u\parallel_{a} v$, and the types of its free
variables $x_{1}^{A_{1}}, \ldots, x_{n}^{A_{n}}$. The Subformula
Property tells us that, no matter what our notion of computation will
turn out to be, when the computation is done, no object whose type is
more complex than the types of the inputs and the output should
appear. If the prime factors of the types $B$ are not subformulas of
$A_{1}, \ldots, A_{n}, A$, then these prime factors should be taken
into account in the complexity measure we are looking for. This leads
to the following definition.

\begin{definition}[Communication Complexity]\label{def:comcompbroad} Let $u \p_a v : A$ a proof
term with free variables $x_{1}^{A_{1}}, \ldots, x_{n}^{A_{n}}$.
Assume that $\send{a}:\NON B$ occurs in $u$ and  $a:B$ occurs in $v$.
\begin{itemize}
\item $B$ is the \textbf{communication kind} of $a$.
\item The \textbf{communication complexity} of $a$ is the maximum
among $0$ and the number of symbols of the prime factors of $B$ that
are neither proper subformulas of $A$ nor strong subformulas of any $A_{1}, \mydots, A_{n}$.
\end{itemize}
\end{definition}

%The chosen condition required to fire a cross reduction for terms
%$u\p_{a} v$ is: 

To fire a cross reduction for $u\p_{a} v$ we require that the
communication complexity of $a$ is greater than $0$.  As this is a
warning that the Subformula Property does not hold, we are using a
logical property as a \emph{computational criterion} for determining
when a computation should start and stop.  To see it at work, consider
again the term $\lambda y^{B}\, a^{\non B} \, y \p_{a} x^{B\impl \non
B}\, a^{B} : \non B$ in the reduction (\ref{examplecomplexity}). Since
all prime factors of the communication kind $B$ of $a$ are proper
subformulas of the type $\non B$ of the term, the communication
complexity of $a$ is $0$ and the cross reduction is not fired, thus
avoiding the loop in (\ref{examplecomplexity}).

% This property detects the the \emph{essential} operations that really
% have to be done.

%Finding good criteria for our calculus was indeed not trivial 
%Actually, unrestricted cross reductions do not always
%terminate. For example, $\lambda x\; \send{a}\, x \p_{a} u\mapsto
%(\lambda x\; \send{b}x\p_{a} u) \p_{b} u[b/a]\mapsto \lambda x\;
%\send{b}x \p_{b} u[b/a]$. On the other hand the Subformula Property
%fares pretty well as a stopping criterion. 
%and the Subformula Property fares pretty well.
%In a sense, it detects all
%the \emph{essential} operations that really have to be done.

Finally, we recall the notion of stack~\cite{Krivine}:
% , known as
% \emph{continuation}, because it
% that embodies 
a series of operations
and arguments.
% tasks that wait to be carried out. 
%A stack represents, from the logical perspective, a series of elimination
%rules; from the $\lambda$-calculus perspective, a series of either operations or arguments.

\begin{definition}[Stack]\label{definitionstack}
A \textbf{stack} is a, possibly empty, sequence \mbox{$\sigma =
\sigma_{1}\sigma_{2}\ldots \sigma_{n} $} such that for every $ 1\leq
i\leq n$, exactly one of the following holds: $\sigma_{i}=t$,
with $t$ proof term or $\sigma_{i}=\pi_{j}$ with $j\in\{0,1\}$, 
% {\ehhi or
% $\sigma_{i}=[x_{1}.u_{1}, x_{2}.v_{2}]$ } 
or $\exfalso_{P}$ for some
atom $P$.  We will denote the \emph{empty sequence} with $\epsilon$
and with $\xi, \xi', \ldots$ the stacks of length $1$. If $t$ is a
proof term, $t\, \sigma$ denotes the term $(((t\,
\sigma_{1})\,\sigma_{2})\ldots \sigma_{n})$.
 \end{definition}

%\begin{definition}[Communication Complexity]\label{def:comcompbroad}
% Let $\pp{a}{u_{1}\p\dots\p u_{m}} : A$ a proof term with free variables $x_{1}^{A_{1}},
%\ldots, x_{n}^{A_{n}}$.  Assume that, for some $1 \leq j \leq m$,
%$\send{a}:\NON B$ occurs in $u_1 \dots u_j$ and $a:B$ in $u_{j+1} \dots
%u_m$.
%\begin{itemize}
%\item $B$ is the \textbf{communication kind} of $a$.
%\item The \textbf{communication complexity} of $a$ is the maximum
%among $0$ and the numbers of symbols of the prime factors of $B $ that
%are neither proper subformulas of $A$ nor strong subformulas of one
%among $A_{1}, \ldots, A_{n}$.
%\end{itemize}
%\end{definition} 

% The idea behind the normalization strategy is to use permutations of
% the $\p$ operator to obtain a term in parallel form and then to employ
% a suitable complexity measure to chose the terms $u \p_a v$ to reduce.
% Since cross reductions can be applied as long as there is a violation
% of the Subformula Property, the natural approach is to define the
% complexity measure as a function of some fixed set of formulas,
% representing the formulas that can be safely used without violating
% the Subformula Property.

We show now that the reductions of the calculus are sound proof
transformations.
\begin{theorem}[Subject Reduction]\label{subjectred}
If $t : A$ and $t \mapsto u$, then $u : A$ and all the free variables of $u$ appear among those of $t$.
\end{theorem} 
\begin{proof} 

It is enough to prove the claim for cross reductions.
%  the theorem for the cross reductions: if $ t : A$
% and $t \mapsto u$, then $u : A$.
 The proof that the intuitionistic
reductions and the permutation rules preserve the type is completely
standard, see e.g.~\cite{Girard}. Basic cross reductions require straightforward considerations as well.
%Suppose then that
%  $ \pp{a}{ {\mathcal C}_1  \p \dots \p {\mathcal C}_i [\send{a} t]  \p \dots \p
%    {\mathcal C}_k \p \dots \p {\mathcal C}_m} \; \mapsto \; \pp{b}{
%    s_{k-1} \p \dots \p s_m } $
%where \[s_{k-1}= \pp{a}{ {\mathcal C}_1  \p \dots \p {\mathcal C}_i
%    [\send{b} \lan \sq{y}\ran ]  \p \dots \p
%    {\mathcal C}_k \p \dots \p {\mathcal C}_m}
%% \qquad \text{ if }  l < k
%\quad \text{ and }\quad s_{l}= {\mathcal C}_l [t ^{b / \lan \sq{y}\ran} / a ] \quad \text{
%if } k \leq l \] where $\send{b}:\NON B$, $b:B$, $\send{a}:\NON A$,
%$a:A$. 
Suppose that \[ (\dots \p {\mathcal C}[\send{a} \, v] \p \dots) \p_{a}
\mathcal{D} \mapsto ( \mathcal{C}[\send{b} \, \lan \sq{y} \ran
] \parallel_{a} \mathcal{D} ) \, \parallel_{b}\, \mathcal{D}[ v^{\, b /
\sq{y} } /a ] \] Since $ \langle \sq{y} \rangle: B $, then ${b}^{\lnot
B} \lan \sq{y}\ran : \FAL$ and ${\mathcal C}[\send{b} \lan \sq{y}\ran]
$ can be assigned the correct type. Since the types of $ v ^{\, b / 
\sq{y}}$ and $ a$ are the same, the term ${\mathcal D} [v ^{\, b /
 \sq{y}} / a ]$ is correctly defined. Finally, since
$\send{a}$ is rightmost in $ {\mathcal C} [\send{a} v]$ and $\sq{y}$
is the sequence of the free variables of $t$ which are bound in $
{\mathcal C} [\send{a} v]$, by Def.~\ref{defimultsubst}, all free
variables of $v ^{\, b /  \sq{y}}$ in ${\mathcal D} [v ^{\, b / 
\sq{y}} / a ]$ are also free in $ {\mathcal C} [\send{a} v]$.
Hence, no new free variable is created during the reduction.
\end{proof}

% bec\^ause of the dualit\'e

\section{The Normalization Theorem}\label{sec:norm}

We prove that every proof term of
$\lama$ reduces in a finite number of steps to a normal form. By
Subject Reduction this implies the normalization for $\NJ + (\emiddle)$ proofs. 
%
%
% The proof
% that we present here improves that for G\"odel logic of~\cite{lics2017} and . %
%
% First, we define a
% strategy for selecting in a given $\lama$ term the subterm to be
% reduced. We remark that the permutations between communications have
% been adopted to simplify the normalization proof, but at the same
% time, they undermine strong normalization, because they enable silly
% loops, like in cut-elimination for sequent calculi. Further
% restrictions of the permutations might be enough to prove strong
% normalization, but we leave this as an open problem.
%
%
% The idea behind the normalization strategy is to use permutations of
% the $\p$ operator to obtain a term in parallel form and then to employ
% a suitable complexity measure to chose the terms $u \p_a v$ to reduce.
% Since cross reductions can be applied as long as there is a violation
% of the Subformula Property, the natural approach is to define the
% complexity measure as a function of some fixed set of formulas,
% representing the formulas that can be safely used without violating
% the Subformula Property.
%
%
%
%
% The normalization proof employs the method introduced
% in~\cite{lics2017}. We adapt it here to $\lama$, improving it and
% filling a gap in the proof of Theorem~\ref{thm:norm} by
% Lemma~\ref{lem:uebermensch}.
%
%
The normalization proof is based on the method introduced
in~\cite{lics2017}, adapted here to $\lama$, refined and
completed by the new Lemma~\ref{lem:uebermensch} which fills a
small gap in the proof for $\lamg$.
The idea behind the normalization strategy is to employ a suitable
complexity measure for terms $u\p_a v$ and, each time a
reduction has to be performed, to choose the term of maximal
complexity. Since cross reductions can be applied as long as there is
a violation of the Subformula Property, the natural approach is to
define the complexity measure as a function of some fixed set of
formulas, representing the formulas that can be safely used without
violating the Subformula Property.
We start by defining  parallel form and normal form. 
\begin{definition}[Parallel Form]\label{definition-parallel-form} 
A \textbf{parallel form} is defined inductively as follows: a simply
typed $\lambda$-term is a parallel form; if $u$ and $v$ are parallel
forms, then both $u\p_{a} v$  and  $u\p v$ are parallel forms.
 %A term $t$ is a
%\textbf{parallel form} if and only if, for any subterm $s \p_a u$ of
%$t$,  $s \p_a u$ is a subterm only of terms of the form $v \p_b w$.
\end{definition}

% \begin{definition}[Parallel
% Form]\label{definition-parallel-form} A term $t$ is a
% \textbf{parallel form} whenever, removing parentheses,
% % and the subscript of $\p$ operators, 
% can be written as
% \[t = t_{1}\parallel_{a_1}
% t_{2}\parallel_{a_i} \ldots \parallel_{a_m} t_{n}\]
% where each
% $t_{i}$, for $1\leq i\leq n$, is a simply typed $\lambda$-term and $a_k$ are channels. 
% \end{definition}
 
 \begin{definition}[Normal Forms and Normalizable Terms]\mbox{}
 \begin{itemize}
\item  A \textbf{redex} is a term $u$ such that $u\mapsto v$ for some $v$ and basic reduction of Figure~\ref{fig:red}. A term $t$ is called a \textbf{normal form} or, simply, \textbf{normal}, if there is no $t'$ such that $t\mapsto t'$. We define $\nf$ to be the set of normal $\lama$-terms. 
\item  
A sequence, finite or infinite, of proof terms
$u_1,u_2,\ldots,u_n,\ldots$ is said to be a reduction of $t$, if
$t=u_1$, and for all  $i$, $u_i \mapsto
u_{i+1}$.
 A proof term $u$ of $\lama$ is  \textbf{normalizable} if there is a finite reduction of $u$ whose last term is a normal form. 
\end{itemize}
\end{definition}

First established in \cite{lics2017}, the following property of simply typed
$\lambda$-terms is crucial for our normalization proof. It ensures that every bound hypothesis appearing
in a normal intuitionistic proof is a strong subformula of one of the
premises or a proper subformula of the conclusion. This
%  property sheds
% light on the complexity of communication reductions: it
implies that the types of the new channels generated by cross reductions
are smaller than the local premises.

\begin{proposition}\label{prop:boundhyp}
Suppose that $t\in\nf$ is a simply typed $\lambda$-term,
$x_{1}^{A_{1}}, \ldots, x_{n}^{A_{n}}\vdash t: A$
and $z: B$ is a variable occurring bound in $t$.  Then one of the following holds:
(1) $B$ is a proper subformula of a prime factor of $A$ or
(2) $B$ is a strong subformula of one among $A_{1},\ldots, A_{n}$.
\end{proposition}

% \begin{proof}
% By induction on $t$. 
% % %%% APPENDIX REFERENCE (removed for Gandalf 2018)
% % See the Appendix.
% \end{proof}

As shown in~\cite{lics2017}, each hypothesis of a
normal intuitionistic proof is followed by an elimination rule, unless
the hypothesis  is $\bot$, subformula of the conclusion or proper subformula of
a premise.

% if the hypothesis is neither $\bot$ nor a
% subformula of the conclusion nor a proper subformula of some other
% premise.

 \begin{proposition}\label{prop:app}
 Let $t\in \nf$ be a simply typed $\lambda$-term and 
$x_{1}^{A_{1}}, \ldots, x_{n}^{A_{n}}, z^{B}\vdash t: A$
 One of the following holds:
 \begin{enumerate}

 \item Every occurrence of $z^{B}$ in $t$ is of the form $z^{B}\, \xi$ for some proof term or projection $\xi$.

 \item $B=\bot$ or $B$ is a subformula of  $A$ or a  proper subformula of one among $A_{1}, \ldots, A_{n}$.
 \end{enumerate}
 \end{proposition}
%  \begin{proof}
%    Easy structural induction on the term. 
%  % %%% APPENDIX REFERENCE (removed for Gandalf 2018)
%  % See the Appendix.
% \end{proof}

\begin{proposition}[Parallel Form Property]
\label{prop:parallelform} 
If $t\in
\nf$ is a $\lama$-term, then it is in parallel form.
\end{proposition}
\begin{proof}
      Easy structural induction on $t$ using the permutation
      reductions. 
% %%% APPENDIX REFERENCE (removed for Gandalf 2018)
% See the Appendix.
\end{proof}

\begin{definition}[Complexity of Parallel
  Terms]\label{def:acomplex}
Let $\mathcal{A}$ be a finite set of formulas. The
$\mathcal{A}$-\textbf{complexity} of $u\p_a v$ 
% $\pp{a}{u_1 \p \dots \p u_m}$
is the  sequence $(c, d, l, o)$ of natural numbers, where:
\begin{enumerate}
\item   if the communication kind of $a$ is $C$, then $c$
is the maximum among $0$ and the number of symbols of the prime
factors of $C$ that are not subformulas of some
formula in $\mathcal{A}$;
\item $d$ is the number of occurrences of $\parallel_e$ and $\parallel$ in $u, v$ for
  any variable $e$;
\item $l$ is the sum of the maximal lengths of the intuitionistic
  reductions of $u,v$; 
% to reach intuitionistic normal form
\item $o$ is the number of occurrences of $a$ in $u,v$.
\end{enumerate} The \textbf{$\mathcal{A}$-communication-complexity}
of
$u \p_a v$ is $c$.
\end{definition}
We adapt the normalization algorithm of~\cite{lics2017} that
represents the constructive content of the proofs of
Prop.~\ref{propositionnormpar} and Thm.~\ref{thm:norm}. Essentially,
the master reduction strategy consists in iterating the basic
reduction relation $\succ$ defined in Def.~\ref{defredstrategy} below, whose goal is to permute the
smallest redex $u\p_a v$ of maximal complexity until $u, v$ are simple
parallel terms (see Def.~\ref{def:simpparterm}) then normalize them
and apply cross reductions.

% , which are
% used for the Normalization Theorem

% define the

% are simply typed $\lambda$-terms

\begin{definition}[Side Reduction Strategy]\label{defredstrategy}
Let $t: A$ be a term with free variables $x_{1}^{A_{1}},\ldots,
x_{n}^{A_{n}}$ and $\mathcal{A}$ be the set of the proper subformulas
of $A$ and the strong subformulas of the formulas $A_{1}, \ldots,
A_{n}$.  Let $u \p_a v$ 
% $\pp{a} {u_1 \p \dots \p u_m } $ 
be the \emph{smallest subterm} of $t$, if
any, among those of \emph{maximal} $\mathcal{A}$-complexity and let
$(c, d, l, o)$ be its $\mathcal{A}$-complexity.  We write $t\succ
t'$ whenever $t'$ has been obtained from $t$ by applying to
$u\parallel_{a} v$:
\begin{enumerate}
\item if $d>0$, a permutation that move $\p_{a}$ inside $u$ or $v$,
  such as  $ u \p_a (v_1 \p_b v_2)
  \mapsto (u \p_a v_1) \p_b (u \p_a  v_2)$
% \begin{align*} & \pp{a} { u_1 \p \dots \p \pp{b}{w_1 \p \dots \p w_n}
% \p \dots \p u_m} \; \mapsto \\ & \pp{b}{ \pp{a}{u_1 \p \dots \p w_1
% \p \dots \p u_m } \p \dots  \p \pp{a}{u_1 \p
% \dots \p w_n \p \dots \p u_m }};
% \end{align*}
\item if $d=0$ and $l>0$, intuitionistic reductions normalizing all terms $u_1 , \dots , u_m$;
\item if $d=l=0$ and $c>0$, a cross reduction possibly followed 
%by intuitionistic reductions normalizing the newly generated simply typed $\lambda$-terms and
by applications of the cross reductions $w_1 \p_c w_2 \mapsto w_i$ for
$i \in \{1, 2\} $ to the whole term;
% $ \pp{b}{v_1 \p \dots \p v_p} \mapsto v_{j_1} \p \dots \p
% v_{j_q} $, for $1 \leq j_1 < \dots < j_q \leq p $ to the whole term.
\item if
$d=l=c=0$, a cross reduction  $u \p_a v \mapsto u$ or $u \p_a v \mapsto v $.
\end{enumerate}
\end{definition}

\begin{definition}[Master Reduction
Strategy]\label{defi-mastredstrategy} We define a normalization
algorithm $\nor{N}(t)$ which for any term $t$ outputs a term $t'$ such that
$t\mapsto^{*} t'$.  Let the free variables of $t$ be
$x_{1}^{A_{1}},\ldots, x_{n}^{A_{n}}$ and $\mathcal{A}$ be the set
of proper subformulas of $A$ and strong subformulas of $A_{1}, \ldots,
A_{n}$. The algorithm behaves as follows.
\begin{enumerate}
\item If $t$ is not in parallel form, then by permutation
reductions $t$ is reduced to a $t'$ which is in parallel form and
$\nor{N}(t')$ is recursively executed.
  
\item \label{dot:simply} If $t$ is a simply typed $\lambda$-term, it
is normalized and returned. If $t= u_1 \p_a u_2 $ is not
a redex, then let $\nor{N}(u_i)=u_i '$ for $1 \leq i \leq 2$. If
$u_1 '\p_a u_2 '$ is normal, it is returned. Otherwise,
$\nor{N}(u_1 '\p_a u_2 ')$ is recursively executed.

\item If $t$ is a redex, 
we select the \emph{smallest}
subterm $w$ of $t$ having maximal
$\mathcal{A}$-communication-complexity $r$. A sequence of terms
$w\succ w_{1}\succ w_{2}\succ \ldots \succ w_{n} $ is produced such
that $w_{n}$ has $\mathcal{A}$-communication-complexity strictly
smaller than $r$. We replace  $w$ by  $w_n$ in $t$, obtain $t'$, and recursively execute $\nor{N}(t')$.

\end{enumerate}  
\noindent We observe that in the step~\ref{dot:simply} of the
algorithm $\nor{N}$, by construction $u_1 \p_a u_2$
is not a redex.  After $u_1, u_2$ are normalized
respectively to $u_1', u_2'$, it can still be the case that
$u_1 ' \p_a u_2 '$ is not normal, because some free
variables of $u_1, u_2$ may disappear during the
normalization, causing a new violation of the Subformula Property that
transforms $u_1 ' \p_a u_2 '$ into a redex, even
though  $u_1 \p_a u_2$ was not.
 \end{definition}

The first step of the normalization consists in reducing the term in parallel form.

\begin{proposition}\label{propositionnormpar}
Let $t: A$ be any term. Then $t\mapsto^{*} t'$, where $t'$ is a parallel form. 
\end{proposition}
\begin{proof}
Easy structural induction on $t$. 
% %%% APPENDIX REFERENCE (removed for Gandalf 2018)
% See the Appendix.
\end{proof}

We now prove that any term in parallel form can be normalized using  the algorithm $\nor{N}$.
\begin{lemma}\label{lem:uebermensch} Let $t: A$ be a term in parallel
form which is not simply typed and $\mathcal{A}$
contain all proper subformulas of $A$ and be closed under
subformulas. Assume that $r > 0$ is the maximum
$\mathcal{A}$-communication-complexity of the subterms of $t$. Assume
that the free variables $x_{1}^{A_{1}},\ldots, x_{n}^{A_{n}}$ of $t$
are such that for every ${i}$, either each strong subformula of
$A_{i}$ is in $ \mathcal{A}$, or each proper prime subformula of
$A_{i}$ is in $ \mathcal{A}$ or has at most $r$ symbols.  Suppose
moreover that no subterm $u_1 \p_a u_2 $ with
$\mathcal{A}$-communication-complexity $r$ contains a subterm of the
same $\mathcal{A}$-communication-complexity. Then there exists $t'$
such that $t \succ^* t' $ and the maximal among the
$\mathcal{A}$-communication-complexity of the subterms of $t'$ is
strictly smaller than $r$.
\end{lemma}
\begin{proof}
We prove the lemma by lexicographic induction on the pair
$(\rho , k)$
where $k$ is the number of subterms of $t$ with
maximal $\mathcal{A}$-complexity $\rho$ among those with
$\mathcal{A}$-communication-complexity $r$.

Let  $u_1 \p_a u_2$ be the \emph{smallest} subterm of $t$ having
$\mathcal{A}$-complexity $\rho$. Four cases can occur.

\noindent \textbf{(a)} $\;$ $\rho =(r, d, l, o)$, with $d>0$. We first show that the term
$u_1 \p_a u_2$ is a redex. Now, the free variables of  $u_1 \p_a u_2$ are among     
$x_{1}^{A_{1}},\ldots, x_{n}^{A_{n}}, a_{1}^{B_{1}},\ldots, a_{p}^{B_{p}}$ and the
communication kind of $a$ is $D$.
Hence, suppose by contradiction that all the prime factors of $D$ are proper subformulas of $A$ or strong subformulas of one among
$A_1, \ldots, A_n, B_1, \ldots, B_p$.
Given that $r>0$ there is a prime factor $P$ of   $D$  such that
$P$ has $r$ symbols and does not belong to $ \mathcal{A}$. The
possible cases are two: (i) $P$ is a proper subformula of a prime
proper subformula $A'_i $ of $A_i$ such that $ A'_i \notin
\mathcal{A}$; (ii) $P$, by Prop.~\ref{propositionstrongsubf},
is a proper subformula of a prime factor of $B_i$. If (i),
then the number of symbols of $A_i'$ is less than or equal to $r$, so
$P$ cannot be a proper subfomula of $A_i'$, which is a
contradiction. If (ii), then, since by hypothesis   { $a_i^{B_i}$  is bound in $t$, there is a prime factor of $B_i$} having a number of symbols greater than $r$, hence we conclude
that there is a subterm $w_1 \p_b w_2$ of $t$ having
$\mathcal{A}$-complexity greater than $\rho$, which is absurd.

Now, since $d>0$, we may assume that for some $1 \leq i \leq 2$,  $u_i
= w_1 \p_b w_2$. Suppose $i=2$.
The term $u_1 \p_a (w_1 \p_b w_2)$
is then a redex of $t$ and by replacing it with $(*)$:  $(u_1 \p_a w_1) \p_b(u_1 \p_a  w_2)$
% \begin{align*}
% \pp{b}{\pp{a}{u_1 \p \dots \p w_1 \p \dots \p u_m } \p \dots\\
%  \p \pp{a}{u_1 \p \dots \p w_q \p \dots \p u_m }}
% \end{align*} 
we obtain from $t$ a term $t'$ such that $t\succ t'$
according to Def.~\ref{defredstrategy}. We must verify that we can
apply to $t'$ the main induction hypothesis. Indeed, the reduction
$t\succ t'$ duplicates all the subterms of $v$, but all of their
$\mathcal{A}$-complexities are smaller than $r$, because $u_1 \p_au_2$ 
by choice is the smallest subterm of $t$ having
maximal $\mathcal{A}$-complexity $\rho$. The terms
$(u_1 \p_a w_i) $ for $1\leq i \leq 2$
have smaller $\mathcal{A}$-complexity than
$\rho$, because they have numbers of occurrences of the symbol
$\parallel$ strictly smaller than in  $u_1 \p_au_2$. 
Moreover, the
terms in $t'$ with $(*)$ as a subterm have, by hypothesis,
$\mathcal{A}$-communication-complexity smaller than $r$ and hence
$\mathcal{A}$-complexity smaller than $\rho$. Assuming that the
communication kind of $b$ is $F$, the prime factors of $F$ that are not in $\mathcal{A}$ must have fewer symbols than the prime
factors of $D$ that are not in $\mathcal{A}$, again because $u_1
\p_au_2$ by choice is the smallest subterm of $t$ having
maximal $\mathcal{A}$-complexity $\rho$; hence, the
$\mathcal{A}$-complexity of $(*)$ is smaller than $\rho$. Hence
the number of subterms of $t'$ with $\mathcal{A}$-complexity $\rho$ is
strictly smaller than $k$. By induction hypothesis, $t'\succ^{*} t''$,
where $t''$ satisfies the thesis.

\noindent \textbf{(b)} $\;$ $\rho=(r, d, l, o)$, with $d=0$ and $l>0$. Since $d=0$, $u_1 , u_2
$ are simple parallel terms
% simply typed $\lambda$-terms
 -- and thus strongly normalizable
\cite{Girard} -- so we may assume that for $1 \leq i \leq 2$, $u_i \mapsto^{*} u_i ' \in \nf$ by a sequence intuitionistic reduction
rules. By replacing in $t$ the subterm $u_1 \p_au_2 $ with  $u_1 '
\p_au_2'  $, we obtain a term $t'$ such that $t\succ t'$
according to Def.~\ref{defredstrategy}. Moreover, the terms in
$t'$ with   $u_1 '
\p_au_2'  $ as a subterm have, by hypothesis,
$\mathcal{A}$-communication-complexity smaller than $r$ and hence
$\mathcal{A}$-complexity is smaller than $\rho$. By induction
hypothesis, $t'\succ^{*} t''$, where $t''$ satisfies the thesis.

\noindent \textbf{(c)}  $\;$  $\rho=(r, d, l, o)$, with $d=l=0$. Since $d=0$, $u_1 , u_2$ are
simply typed $\lambda$-terms. Since $l=0$, $u_1 ,u_2$ are in normal
form and thus satisfy conditions 1.\ and 2.\ of
Prop.~\ref{prop:boundhyp}. We need to check that   $u_1 
\p_au_2  $is a redex, in particular that the
communication complexity of $a$ is greater than $0$. Assume that the
free variables of   $u_1 
\p_au_2 $ are among   { 
$x_{1}^{A_{1}},\ldots,
x_{n}^{A_{n}}, a_{1}^{B_{1}},\ldots,
a_{p}^{B_{p}}$ and that the communication kind of $a$
is $D$. As we argued above, we obtain that not all the
prime factors of $D$ are proper subformulas of $A$ or
strong subformulas of one among $A_1, \ldots, A_n, B_1,
\ldots, B_p$. By Def.~\ref{def:comcompbroad}, $u_1 
\p_au_2 $ is a redex.}

We now prove that every occurrence of $a$ in $u_1 ,u_2$ is of
the form $a\, \xi$ for some term or projection $\xi$. First of all,
$a$ occurs with arrow type in all $u_1,u_2$. Moreover, $u_1
:A , u_2:A$, since $t: A$ and $t$ is a parallel form; hence, the
types of the occurrences of $a$ in $u_1 , u_2$ cannot be
subformulas of $A$, otherwise $r=0$,   and cannot be proper subformulas
of one among $A_{1}, \ldots, A_{n}, B_{1}, \ldots,
B_{p}$, otherwise the prime factors of $D$ would be strong subformulas of one among $A_1, \ldots, A_n,
B_1, \ldots, B_p$ and thus we are done.
Thus by Prop.~\ref{prop:app} we are done.
 Two cases can occur.

\noindent $\bullet$  $a$ does not occur in  $u_i$ for $1 \leq i \leq 2$. 
% some among $u_1 ,u_2$.
%  at least in $u_{j_1} , \dots ,
% u_{j_n}$ for $1 \leq j_1 < \dots < j_n \leq m $. 
By performing a cross reduction, we replace in $t$ the term $u_1 \p_a u_2$
with $ u_i $ and obtain a term $t'$ such that
$t\succ t'$ according to Def.~\ref{defredstrategy}. After the
replacement, the number of subterms having maximal
$\mathcal{A}$-complexity $\rho$ in $t'$ is strictly smaller than the
number of such subterms in $t$. By induction hypothesis, $t'\succ^{*}
t''$, where $t''$ satisfies the thesis.

\noindent $\bullet$  $a$ occurs in all the subterms $u_1 , u_2$.  Let $u_1 = (\dots
\p {\mathcal C} [\send{a} \,w ]\p \dots )$ where $\send{a}:\NON D$, $(\dots \p {\mathcal C} [\send{a} \,w ]\p
\dots )$ is a normal simple parallel term, ${\mathcal C} [\; ]$ is a simple context,
 %  $1 \leq i
% \leq q \leq m$ where $u_1 , \dots , u_{q-1} $ are all terms in which $a$
% occurs with the form $\send{a}$
and the displayed occurrence of
$\send{a}$ is  rightmost in ${\mathcal C} [\send{a} \,w ]$. 
%and $\sigma _i $ is the stack of \emph{all} terms or projections $a$ is applied to. 
By applying a cross reduction to ${\mathcal C} [\send{a}\,w ] \p_a u_2$
% $$ \pp{a}{{\mathcal C}_1 [a\,w_1 ] \p \dots \p
% {\mathcal C}_{q-1}[a\,w_{q-1} ] \p u_{q} \p \dots \p
% u_m } $$ 
% communicating the term $w_z \in \{ w_1 , \dots , w_q\}$  
we obtain either the term
$ u_2 [w / a ]$
or the term $(\ast\ast)$ 
%\[(\ast\ast) \qquad \pp{b}{ s_q \p \dots \p s_m }\]  where $s_q$
% for $1 \leq i \leq q $ 
%is
 $({\mathcal C} [\send{b} \lan \sq{y}\ran ] \p_a u_2 ) \p_b u_2 [w^{\,
   b  /  \sq{y}}  / a ]$
% \[ \pp{b}{\pp{a}{  {\mathcal C}_z [\send{b} \lan \sq{y}\ran ]  \p
   % u_{q} \p \dots \p u_m} \p u_{q} [w_z^{b  /  \sq{y}} / a ]\p\dots\p u_{m} [w_z^{b  /  \sq{y}} / a ]    }
   %  \]
    where $\send{b}:\NON B$, $\sq{y}$ is the sequence of the free variables of $w_z$
which are bound in $\mathcal{C}[\send{a}\, w]$ and $\send{a}$ does
not occur in $w$. 
   In the former case, the term has $\mathcal{A}$-complexity strictly smaller than $\rho$ and we are done. In the latter case, 
% \[  \pp{a}{{\mathcal C}_1 [a\,w_1 \, \sigma _1 ] \p \dots \p
% {\mathcal C}_q[a\,w_q\, \sigma _q ] \p {\mathcal C}_{q+1} \p \dots \p
% {\mathcal C}_m } 
% \]
%and for $q+1 \leq i \leq m$ is
%\[{\mathcal C}_i [w_z^{b  /  \sq{y}} / a ]\]
% \[\pp{a}{{\mathcal C}_1 [a\, w_1 \, \sigma _1] \p \dots \p {\mathcal C
% } _i[ w_j^{b^{B_i \impl B_j} \lan \sq{y}_i \ran / \sq{y}_j} \, \sigma
%  _i ] \p \dots \p {\mathcal C}_m[a\,w_m \, \sigma _m]}
% \]
 since $u_1 , u_2$ satisfy conditions 1.\
and 2.\ of Prop.~\ref{prop:boundhyp}, the types $Y_1 ,
\dots , Y_k$ of the variables $\sq{y}$ are proper subformulas of $A$ or
strong subformulas of the formulas $A_{1}, \ldots, A_{n}, B_{1},
\ldots, B_{p}$. Hence, the types among $Y_{1}, \ldots, Y_{k}$ which
are not in $\mathcal{A}$ are strictly smaller than all the prime
factors of the formulas $B_1, \ldots , B_p$.  Since the communication
kind of $b$ consists of the formulas $Y_{1}\et \ldots \et Y_{k}$, by
Def.~\ref{def:acomplex}  the $\mathcal{A}$-complexity
of the term $(\ast\ast)$ above is strictly smaller than the
$\mathcal{A}$-complexity $\rho$ of $u_1 \p_a u_2$.
%or the communication kind of $b$ is $\top$.  In the latter case we apply
%a cross reduction of the shape $ \pp{b}{s_1 \p \dots \p s_m} \mapsto s_{j_1} \p
%\dots \p s_{j_l} $, for $1\leq j_1 < \dots < j_l \leq m $ and obtain a
%term with $\mathcal{A}$-complexity strictly smaller than $\rho$.

%In the former case, 
%let $w_z'$ and $v$ be simply typed $\lambda$-terms such
%that
%\[w_z^{b\lan \sq{y} \ran } \mapsto^{*}
%w_z' \in\nf \qquad \text{ and } \qquad b^{\NON B} \lan \sq{y}\ran \sigma_z \mapsto^{*}
%v \in\nf \]
%Let now $u_{j}'$ be the normal form of $u_{j}[w_z^{b  / \lan \sq{y}\ran} / a ]$, for $q\leq j\leq m$. 	
%By hypothesis, $a$ does not occur in $w_z' $
%and thus neither in $v$. 
%Moreover, by the assumptions on $\sigma_z$
%and since 
Now, since $ {\mathcal C} [\send{a}\,w ] , u_2   $  normal simple parallel terms, $
{\mathcal C} [\send{b} \lan \sq{y}\ran ]$ 
% for $1 \leq i \leq q $
is normal too and contain fewer occurrences of $\send{a}$ than ${\mathcal C} [\send{a}\,w]$ does; hence, the
$\mathcal{A}$-complexity of the term
$\quad {\mathcal C} [\send{b} \lan \sq{y}\ran ] \p_a u_2 \quad $
% for $1 \leq i \leq q $ 
%and
%\[{\mathcal C}_i [w_z' / a ]\]
% for $q+1 \leq i \leq m$ 
%
% \[\pp{a}{{\mathcal C}_1 [a\, w_1 \, \sigma _1] \p \dots \p {\mathcal
% C } _i[ w_j' ] \p \dots \p {\mathcal C}_m[a\,w_m \, \sigma _m]} \] 
is
strictly smaller than the $\mathcal{A}$-complexity $\rho$ of $u_1 \p_a
u_2$. Let now $t'$ be the term obtained from $t$ by
replacing the term ${\mathcal C} [\send{a}\,w ] \p_a u_2$
%   $  \pp{a}{{\mathcal C}_1 [\send{a}\,w_1 ] \p \dots \p
% {\mathcal C}_{q-1}[\send{a}\,w_{q-1} ] \p u_{q} \p \dots \p
% u_m }  $  
with $(\ast \ast)$.
%  $\pp{a}{ {\mathcal C}_1 [a\,w_1 \, \sigma _1 ] \p
%% \dots \p {\mathcal C}_1 [a\,w_1 \, \sigma _1 ] } $
%with
%\[(\ast\ast\ast) \qquad\pp{b}{\pp{a}{  {\mathcal C}_z [b^{\NON B} \lan \sq{y}\ran ]  \p
%   u_{q} \p \dots \p u_m} \p u_{q}' \p\dots\p u_{m}'  }
%    \]
%where  $s_q'$ is
% for $1
% \leq i \leq m $ is
%\[\pp{a}{ {\mathcal C}_1 [a\,w_1 \, \sigma _1 ] \p \dots \p {\mathcal C}_z [v ]  \p \dots \p
%{\mathcal C}_q[a\,w_q\, \sigma _q ] \p \dots \p
%    {\mathcal C}_{q+1} \p \dots \p {\mathcal C}_m}\]
%and $s_i'$ for $q+1 \leq i \leq m$ is
%\[{\mathcal C}_i [w_z' / a ]\]
% \[\pp{a}{{\mathcal C}_1 [a\, w_1 \, \sigma _1] \p \dots \p
% {\mathcal C } _i[ w_j' ] \p
% \dots \p {\mathcal C}_m[a\,w_m \, \sigma _m]} \]
By construction $t\succ t'$.  Moreover, the terms in $t'$
with \mbox{$(\ast\ast)$} as a subterm have, by hypothesis,
$\mathcal{A}$-communication-complexity smaller than $r$ and hence
$\mathcal{A}$-complexity  smaller than $\rho$. Hence, we can apply
the main induction hypothesis to $t'$ and obtain by induction
hypothesis, $t'\succ^{*} t''$, where $t''$ satisfies the thesis.

\noindent \textbf{(d)} $\;$  $\rho =(r, d, l, o)$, with $d=l=o=0$. Since $o=0$, $u_1\p_a u_2$
is a redex. Let us say $a$ does not occur
in $u_i$ for $1 \leq i \leq 2$. By performing a cross reduction, we
replace $u_1 \p_au_2$ with $ u_i $ according to
Def.~\ref{defredstrategy}. Hence, by induction hypothesis,
$t'\succ^{*} t''$, where $t''$ satisfies the thesis.
\end{proof}

\begin{theorem}\label{thm:norm} Let  $ t: A$ be a $\lama$-term. 
  Then $t\mapsto^{*} t': A$, where $t'$ is a normal parallel
  form.
\end{theorem}
% \begin{proposition}\label{proposition-normcomp} Let $t: A$ be any term
% in parallel form. Then $t\mapsto^{*} t'$, where $t'$ is a
% parallel normal form.
% \end{proposition}
\begin{proof} 
% We provide here a sketch of the proof, which can be found in the Appendix.
% %%% APPENDIX REFERENCE
By Prop.~\ref{propositionnormpar}, we can assume that $t:A$ is in parallel form.
Assume now that the free variables of $t$ are
$x_{1}^{A_{1}},\ldots, x_{n}^{A_{n}}$ and let $\mathcal{A}$ be the set
of the proper subformulas of $A$ and the strong subformulas of the
formulas $A_{1}, \ldots, A_{n}$.  We prove the theorem by
lexicographic induction on the quadruple $(|\mathcal{A}|, r, k, s)$ where $|\mathcal{A}|$ is the cardinality
of $\mathcal{A}$, $r$ is the maximal
$\mathcal{A}$-communication-complexity of the subterms of $t$, $k$ is
the number of subterms of $t$ having maximal
$\mathcal{A}$-communication-complexity $r$ and $s$ is the size of
$t$. If $t$ is a simply typed $\lambda$-term, it has a normal form
\cite{Girard} and we are done; so we assume $t$ is not. There are two
main cases.

\noindent $\bullet$ \textit{First case}: $t$ \emph{is not a redex}.  
Let
$t= u_1 \p_a u_2$ and let $C$ be the communication kind of $a$.
Then, the communication complexity of $a$ is $0$ and by
Def.~\ref{def:comcompbroad} every prime factor of $C$ 
belongs to $\mathcal{A}$. Let the types of the occurrences of $a$
in $u_i$ for $1 \leq i \leq 2 $ be $B_i$, with $B_{i}=\lnot C$ or $B_{i}=C$. Let now  $\mathcal{A}_i$ be the set of the proper
subformulas of $A$ and the strong subformulas of $A_{1},
\ldots, A_{n}, B_i$. By Prop.~\ref{propositionstrongsubf},
every strong subformula of $B_i$ is a proper subformula of a prime factor
of $C$, and this prime factor is in
$\mathcal{A}$. Hence, $\mathcal{A}_i\subseteq \mathcal{A}$.

If $\mathcal{A}_i= \mathcal{A}$, the maximal
$\mathcal{A}_i$-communication-complexity of the terms of $u_i$ is less
than or equal to $r$ and the number of terms with maximal
$\mathcal{A}_i$-communication-complexity is less than or equal to $k$;
since the size of $u_i$ is strictly smaller than that of $t$, by
induction hypothesis $u_i\mapsto^{*} u_i'$, where $u_i'$ is a normal
parallel form.

If $\mathcal{A}_i\subset\mathcal{A}$, again by induction hypothesis
$u_i\mapsto^{*} u_i'$, where $u_i'$ is a normal parallel form.

Let now $t'= u_1 ' \p_a u_2' $, so that $t\mapsto^{*}
t'$. If $t'$ is normal, we are done. Otherwise, since
$u_j'$ for $1\leq j \leq 2 $ are normal, the only possible redex remaining in $t'$ is the
whole term itself, i.e., $u_1 ' \p_a u_2'$: this
happens only if the free variables of $t'$ are fewer than those of
$t$; w.l.o.g., assume they are $x_{1}^{A_{1}}, \ldots, x_{i}^{A_{i}}$,
with $i< n$.  Let $\mathcal{B}$ be the set of the proper subformulas
of $A$ and the strong subformulas of the formulas $A_{1}, \ldots,
A_{i}$. Since $t'$ is a redex, the communication complexity of $a$ is
greater than $0$; by def.~\ref{def:comcompbroad},  a
prime factor of $C$ is not in $\mathcal{B}$, so we have
$\mathcal{B}\subset \mathcal{A}$. By I.H.,
$t'\mapsto^{*}t''$, where $t'' $ is a parallel normal form.

\noindent $\bullet$ \textit{Second case}: $t$ \emph{is a redex}.  Let $u_1 \p_a u_2$ be the \emph{smallest} subterm of $t$ having
$\mathcal{A}$-communication-complexity $r$. The free variables of
$u_1 \p_a u_2 $ satisfy the hypotheses of
Lem.~\ref{lem:uebermensch} either because have type $A_i$ and
$\mathcal{A}$ contains all the strong subformulas of $A_i$, or because
the prime proper subformulas of their type have at most $r$ symbols,
by maximality of $r$. By Lem.~\ref{lem:uebermensch} $u_1 \p_a u_2 \succ^* w $ where the maximal among the
$\mathcal{A}$-communication-complexity of the subterms of $w$ is
strictly smaller than $r$. Let $t'$ be the term obtained replacing $w$
for $u_1 \p_a u_2$ in $t$. We apply the
I.H.\ and obtain $t' \mapsto^*t''$ with $t''$ in
parallel normal form.
\end{proof}

% The normalization for $\lama$, and thus for the natural deduction calculus, easily follows.

% \begin{theorem}\label{thm:norm} Suppose that  $ t: A$ is a term  of $\lama$. Then $t\mapsto^{*} t': A$, where $t'$ is a normal parallel form.
% \end{theorem}

% \section{The Subformula Property}

% We show that normal $\lama$-terms satisfy the important Subformula
% Property: a normal proof does not contain concepts that do not already
% appear in the premises or in the conclusion. This, in turn, implies
% that our Curry--Howard correspondence for $\lama$ is meaningful from
% the logical perspective and produces analytic natural deduction
% proofs.

We prove now that the Subformula Property holds: a normal proof does
not contain concepts that do not already appear in the premises or in
the conclusion.

% A proof of a normal $\lama$-term does not contain concepts that do not
% already appear in the premises or in the conclusion.

\begin{theorem}[Subformula Property]\label{thm:subf}
Suppose $x_{1}^{A_{1}}, \ldots, x_{n}^{A_{n}}\vdash t: A $, with $t\in
\nf$. Then
% By induction on $t$ is easy to prove that
% \[x_{1}^{A_{1}}, \ldots, x_{n}^{A_{n}}\vdash t: A \quad \mbox{and} \quad t\in \nf. \quad \mbox{Then}:\]
$(i)$ for each communication variable $a$ occurring bound in $t$ and
with communication kind $C_1 , \dots , C_m$, the prime factors of $C_1
, \dots , C_m$ are proper subformulas of $A_{1}, \ldots, A_{n}, A$;
$(ii)$ the type of any subterm of $t$ which is not a bound
communication variable is either a subformula or a conjunction of
subformulas of $A_{1}, \ldots, A_{n}, A$.
\end{theorem}
\begin{proof}
By induction on $t$. 
% %%% APPENDIX REFERENCE (removed for Gandalf 2018)
% See the Appendix.
\end{proof}

\section{On the expressive power of $\lamem$}
\label{Sec:exprpower}
We discuss the relative expressive power of $\lamem$ and its computational capabilities.

\paragraph{Comparison with $\pi$-calculus and $\lambda _G$}
\label{comparison}
%{\bf Comparison with $\pi$-calculus:}
%To better explain some of the features of $\lamem$ we compare it below with the $\pi$-calculus.
%Recall however their different scope: $\lamem$ is a parallel functional language
%while $\pi$-calculus a formalism for {\em modeling} concurrent systems.
%
In contrast with the $\pi$-calculus~\cite{Milner, sangiorgiwalker2003}
which is a formalism for {\em modeling} concurrent systems, $\lamem$ is a parallel {\em functional language} 
intended (as a base) for programming.
The first similarity between the two calculi is in the channel restrictions: the $a$ of $u
\p_a v $ in $\lama$ and of $\nu a (P \mid Q )$ in
$\pi$-calculus have the same r\^ole. Moreover the result of communicating in $\lamem$
a closed process or data is as in the asynchronous $\pi$-calculus~\cite{HT}.
%Nonetheless, subtle, but relevant, technical differences separate the communications mechanism
%of the two calculi.
In contrast with the $\pi$-calculus whose communication only applies to data or channels,
the communication in $\lama$ is higher-order. Moreover, the latter
%The main difference between $\pi$-calculus and $\lama$ is that the
%former only supports data or channel communication, while the latter supports
%higher-order communication.  The higher-order communication of
%$\lama$ is also very general and 
can handle not only closed and
open processes, but also processes that are closed in their original
environment, but become open after the communication.
The number of recipients of a communication is another difference
between $\pi$-calculus and
$\lama$. While in pure $\pi$-calculus both sender and recipient of a
communication might be selected non-deterministically, in
$\lama$, since the communication is a broadcasting to all recipients,
only the sender, ${\mathcal C}_i[\send{a} \, t_i]$ for $i \in \{1
,\dots , n \}$ in the following example,
can be non-deterministically selected: \[({\mathcal C}_1[\send{a} \, t_1] \p \dots \p {\mathcal C}_n
[\send{a} \, t_n]) \p_a {\mathcal D} \quad  \mapsto \quad  {\mathcal
D}[t_i/a]\]
Furthermore, in $\pi$-calculus only one process can receive each message
whereas in $\lama$ we can have one-to-many communications, or broadcast:
\[{\mathcal C}[\send{a} \, u] \p_a( {\mathcal D}_1 \p \dots \p
{\mathcal D}_m) \quad \mapsto \quad  {\mathcal D}_1[u/a] \p\dots \p
{\mathcal D}_m [u/a]\]  Finally, while in $\pi$-calculus there is
no restriction on the use of channels between processes, in $\lama$
there are strict symmetry conditions; similar conditions are adopted
in typed versions of $\pi$-calculus (see \cite{TCP2013,
Wadler2012}). Hence $\lama$ cannot encode a \emph{dialogue} between
two processes: if a process $u$ receives a message from a process $v$,
then $v$ cannot send a message to $u$. To model these exchanges, more
complex calculi such as $\lambda _G$~\cite{lics2017} are
needed. If a $\lamem$ channel connects two processes as shown below on the left, a
$\lambda _G$ channel connects them as shown below on the right.\\
\begin{figure}[h!]  \centering \tikzstyle{proc}=[circle, minimum
size=3mm, inner sep=0pt, draw]
  \begin{tikzpicture}[node distance=1.5cm,auto,>=latex', scale=1.5]
\node [proc] (1) at (-3.3,0) {}; \node [proc] (2) at (-1.3,0) {};
\path[->] (1) edge [thick, bend left] (2); \node [proc] (3) at (1.3,0)
{}; \node [proc] (4) at (3.3,0) {}; \path[<->] (3) edge [thick, bend
left] (4);
\end{tikzpicture}
\end{figure}\\ Namely, a $\lamg$ channel can transmit messages between
the processes in both directions. Even though the communication
mechanism of $\lamg$ enables us to define unidirectional channels as
well, the technical details of $\lamg$ and $\lamem$ communications
differ considerably since they are tailored, respectively, to the
linearity axiom $(A\impl B) \vel (B \impl A)$ and to $\emiddle$.
% For instance, the parallel OR term in the proof of
% Proposition~\ref{power} is simpler than the analogous term defined in
% $\lamg$~\cite{lics2017}.
In $\lamem$ indeed all occurrences of the
receiver's channel are simultaneously replaced by the message, but
this is not possible in $\lamg$. As a consequence, $\lamg$ cannot
implement broadcast communication.  Finally, while the closure
transmission mechanisms of $\lamg$ and $\lamem$ have the same function
and capabilities -- a version of Example~\ref{ex:optimize} for $\lamg$
is presented in~\cite{lics2017} -- $\lamem$ mechanism is considerably
simpler.

% due to the shape of $\emiddle$ .

We establish first the relation of $\lamem$ with the simply
typed $\lam$-calculus and Parigot's $\lam _\mu$~\cite{Parigot},
by proving in particular that $\lamem$, as $\lamg$, can code the parallel OR. 
Then we show the use of $\lamem$ closure transmission for code
optimization. % Both are similarly established in~\cite{lics2017} for $\lamg$.

\begin{proposition} \label{power} $\lamem$ is strictly more
  expressive than simply typed $\lam$-calculus and propositional
  $\lam _\mu$.
\end{proposition}
\begin{proof} 
The simply typed $\lambda$-calculus can be trivially embedded into $\lamem$. The converse
does not hold, as $\lamem$ can encode
the parallel OR, which is a term $\mathsf{O}: \mathsf{Bool}
\IMPL \mathsf{Bool} \IMPL\mathsf{Bool} $ such that
$\mathsf{O}\mathsf{F}\mathsf{F} \mapsto^* \mathsf{F}$,
$\mathsf{O}u\mathsf{T} \mapsto^* \mathsf{T}$, $\mathsf{O}\mathsf{T}u
\mapsto^* \mathsf{T}$ for every term $u$.
By contrast, as a consequence of Berry's sequentiality theorem
(see~\cite{Barendregt})
there is no parallel OR in simply typed $\lam$-calculus. 
% in order to evaluate the whole expression.
%$\mathsf{O}$ can instead be defined in Boudol's
%calculus~\cite{Boudol89} and in $\lamg$~\cite{lics2017}. 
Assuming to add the boolean type in our calculus, that  $\send{a}:
\mathsf{Bool}  \et  S  \IMPL \bot$, and that 
$a: \mathsf{Bool}  \et  S   $, ``$\mathsf{if} \, \_ \, \mathsf{then}
\, \_ \, \mathsf{else} \, \_$'' is as usual, the $\lamem$ term for such
parallel OR is
%(as usual the term ``$\mathsf{if} \, u \, \mathsf{then} \, s \,
%\mathsf{else} \, t$'' reduces to $s$ if
%$u = \mathsf{T}$, and to $t$ if $u = \mathsf{F}$):
%
%
% \[\mathsf{O} := \lambda x^{\mathsf{Bool}} \, \lambda y^{\mathsf{Bool}}
% (\mathsf{if} \, x \, \mathsf{then} \, \mathsf{T} \,
%     \mathsf{else} \; \efq{\mathsf{Bool}}{\send{a} \lan \mathsf{F},
%       s\ran } \parallel_a \mathsf{if} \, y \,
%     \mathsf{then} \, \mathsf{T} \, \mathsf{else} \, a \pi _0) \] where  $\send{a}: (\mathsf{Bool} \ET S)  \IMPL \bot$,
% $a: \mathsf{Bool} \ET S $, $S$ is any type, ``$\mathsf{if} \, \_ \, \mathsf{then}
% \, \_ \, \mathsf{else} \, \_$'' is as usual and $s:S$ is any  term.
%
\[\mathsf{O} := \lambda x^{\mathsf{Bool}} \, \lambda y^{\mathsf{Bool}}
(\mathsf{if} \; x \; \mathsf{then} \; \mathsf{T} \;
    \mathsf{else} \; \efq{\mathsf{Bool}} {\send{a}   \lan
      \mathsf{F}  ,  s \ran }  \parallel_a \mathsf{if} \; y \;
    \mathsf{then} \; \mathsf{T} \; \mathsf{else} \; a  \, \pi_0) \]
 for any flag term $s:S$, introduced for a complete control of the
 reduction. Now, $\mathsf{O}\, u \,\mathsf{T}$ reduces to $\mathsf{T}$ by
 \begin{footnotesize}
   \begin{align*} % \mathsf{O}\, u \,\mathsf{T} &\mapsto^{*}
(\mathsf{if} \; u \; \mathsf{then} \; \mathsf{T} \; \mathsf{else} \;
\efq{\mathsf{Bool}}{\send{a}  \lan
      \mathsf{F}  ,  s \ran})
                                       \parallel_{a} (\mathsf{if} \;
\mathsf{T} \; \mathsf{then} \; \mathsf{T} \; \mathsf{else} \; a  \, \pi_0
) \mapsto^{*} (\mathsf{if} \; u \; \mathsf{then} \; \mathsf{T}
\; \mathsf{else} \; \efq{\mathsf{Bool}}{\send{a}  \lan
      \mathsf{F}  ,  s \ran}) \parallel_{a} \mathsf{T} \; \mapsto \; \mathsf{T}
   \end{align*}
 \end{footnotesize}
And symmetrically $\mathsf{O}\,  \mathsf{T}\, u \,
\mapsto^{*} \, \mathsf{T} $. 
On the other hand, $\mathsf{O} \, \mathsf{F} \,\mathsf{F} $ reduces to $\mathsf{F}$ by 
\begin{small}
  \begin{align*}
    % \mathsf{O} \, \mathsf{F} \,\mathsf{F} & \mapsto^{*}
                                            (\mathsf{if} \; \mathsf{F} \; \mathsf{then} \; \mathsf{F} \; \mathsf{else}
                                            \;  \efq{\mathsf{Bool}}{\send{a}   \lan
      \mathsf{F} {  ,  s \ran}})    \parallel_{a} (\mathsf{if} \; \mathsf{F} \;
                                            \mathsf{then} \;
                                              \mathsf{F} \;
                                              \mathsf{else} \; a {  \, \pi_0} )  \mapsto^{*}  \efq{\mathsf{Bool}}{\send{a}    {  \lan }
      \mathsf{F} {  ,  s \ran} }
                                                                                                          \parallel_{a}
a {  \, \pi_0 \; \mapsto \;    \lan \mathsf{F}  ,  s \ran \, \pi_0} \; \mapsto \;    \mathsf{F}
  \end{align*}
\end{small} The claim follows by Ong's embedding of propositional $\lam
_\mu$ in the simply typed $\lam$-calculus, see Lemma 6.3.7
of~\cite{Sorensen}. Indeed the  translation $\underline{u}$ of
a $\lam\mu$-term $u$ is such that $\underline{s\, t} = \underline{s}\;
\underline{t}$ and $\underline{x} = x$ for any variable $x$, if there
were a typed $\lam\mu$-term $\mathsf{O}$ for parallel OR, then
\begin{footnotesize}
 \[   \underline{\mathsf{O}} \,x
    \, \mathsf{T} \; =\;  \underline{\mathsf{O} \,x \, \mathsf{T}}
    \quad   \mapsto^* \quad \underline{\mathsf{T}} \; = \;
             \mathsf{T} \qquad\quad    
    \underline{\mathsf{O}} \,
    \mathsf{T} \, x  \; =\;  \underline{\mathsf{O} \, \mathsf{T} \,
    x}\quad \mapsto^* \quad  \underline{\mathsf{T}} \; = \;   \mathsf{T}
    \qquad\quad   \underline{\mathsf{O}} \,
    \mathsf{F} \, \mathsf{F} \; =\;  \underline{\mathsf{O} \,
    \mathsf{F} \, \mathsf{F}} \quad  \mapsto^* \quad \underline{\mathsf{F}} \; = \;  \mathsf{F} 
\]
\end{footnotesize}
and $\underline{\mathsf{O}}$ would be a parallel OR in simply typed
$\lam$-calculus, which is impossible.
\end{proof}

%Notice that the analogous of Proposition~\ref{power} is similarly
%established in~\cite{lics2017} for $\lamg$.

% see Section~2.4 of~\cite{Parigot} and Section 7 of~\cite{deGrooteex},
% respectively.

% \begin{remark}
% The parallel OR can also be defined in Boudol's calculus~\cite{Boudol89} and in $\lamg$~\cite{lics2017}.
% \end{remark}

\begin{example}[\textbf{Classical Disjunction}] Since in classical
logic disjunction is definable, the corresponding computational
constructs of case distinction and injection can be defined in
$\lamem$.  By contrast, these constructs are usually added as new
primitives in simply typed $\lambda$-calculus, as simulating them
requires complicated CPS-translations.

% where their simulation
% requires complicated CPS-translations.

%These constructs are usually added as new
%primitives in simply typed $\lambda$-calculus, as simulating them
%requires complicated CPS-translations.

The $\lamem$  terms $\inj_0(u)$, $\inj_1(u)$ and $t[x_0.v_0, x_1.v_1]$ such
that for $i \in \{0,1\}$ we have $ \; \inj_i(u)[x_0.v_0, x_1.v_1] \; \mapsto\;
v_i [ u / x_i ] $ are defined as follows: Let $A \vel B  := (A \impl \fal)\impl(
  B \impl \fal )\impl \fal$
%We define That is it is possible to
%define direcly terms $\inj_0(u)$, $\inj_1(u)$ and $t[x_0.v_0, x_1.v_1]$ such
%that for $i \in \{0,1\}$ we have $ \; \inj_i(u)[x_0.v_0, x_1.v_1] \; \mapsto\;
%v_i [ u / x_i ] $. 
\vspace{-4pt}
\begin{gather*}
%\hspace{-150pt}\text{Let}  \qquad\qquad   A \vel B  := (A \impl \fal)\impl(
%  B \impl \fal )\impl \fal  \\
\inj_0(u) \; :=\; \lam x^{A \impl \fal} \lam y^{B \impl \fal}\,
    x\,u : A \vel B \qquad\qquad \inj_1(u)\; :=\; \lam x^{A \impl
      \fal} \lam y^{B \impl \fal}\, y\,u : A \vel B  \\ 
    t\, [x_0.v_0, x_1.v_1] \;  := \; ( \efq{F}{}(t \, \send{a}\,
      \send{b}) \p_a v_0 [a/x_0] ) \p_b v_1 [b/x_1] :F\end{gather*}where
  $a:A,\, \send{a}:A \impl \fal , \, b:B\, , \send{b}:B \impl \fal,\,
  v_0:F, \, v_1:F,\, t:A \vel B$ and $\efq{F}{} $ is a closed term of
  type $\fal \impl F$.
We can then verify, for example, that
\begin{footnotesize}
  \[\inj_0 (u) \, [x_0.v_0, x_1.v_1] \; := \; ( \efq{F}{}((\lam x^{A
\impl \fal \lam y^{B \impl \fal}\, x\,u) \, \send{a}\, \send{b})} \p_a
v_0 [a/x_0] ) \p_b v_1 [b/x_1] \]
  \[\mapsto^*\; ( \efq{F}{}(\send{a}\, u) \p_a v_0 [a/x_0] ) \p_b v_1
[b/x_1] \; \mapsto\; v_0 [u/x_0] \p_b v_1 [b/x_1] \; \mapsto\; v_0
[u/x_0] \]
\end{footnotesize}
%%% $\send{a} $ becomes $\lam z \, \send{a} \lan huge_type_term(F) ,z
%%% \ran  $ in the right hand side of the $\p$ we use $a\pi_1$  
\end{example}

\begin{example}[\textbf{Cross reductions for program efficiency}]
\label{ex:optimize} 
%% OLD EXPLANATION
% Consider now the parallel processes: $M \parallel_{a} (
% Q \parallel_{b} P)$. The process $Q$ contains a channel $b$ to send
% the message $\lan s , y \ran$ to $P$. But the variable $y$ stands for a
% missing part of the message which needs to be replaced by a term that
% $M$ has to compute and send to $Q$. Hence the whole interaction needs to 
% wait for $M$. The cross reduction
% handles precisely this kind of missing arguments. It
% enables $Q$ to send immediately the message through the channel $b$
% and establishes a new communication channel on
% the fly which redirects the missing term, when ready, to the new location of the message inside $P$.

We show how to use cross reductions 
to communicate
processes that are still waiting for some arguments.
%
%; in other terms to use  as a mechanism to handle the closures of higher-order
%messages. %An equivalent problem is discussed in~\cite{EBPJ2011}.
%
Consider the process $M \parallel_{a} (
Q \parallel_{b} P)$. The process $Q$
contains a channel $b$ to send a message (yellow pentagon) to $P$
(below left), but
the message is missing a part (yellow square) which is computed by $M$
and sent to $Q$ by $a$. In a system without a closure handling
mechanism, the whole interaction needs to wait until $M$ can
communicate  to $Q$ (below right).
\begin{center}
\vspace{-4pt}
\includegraphics[width=0.4\textwidth]{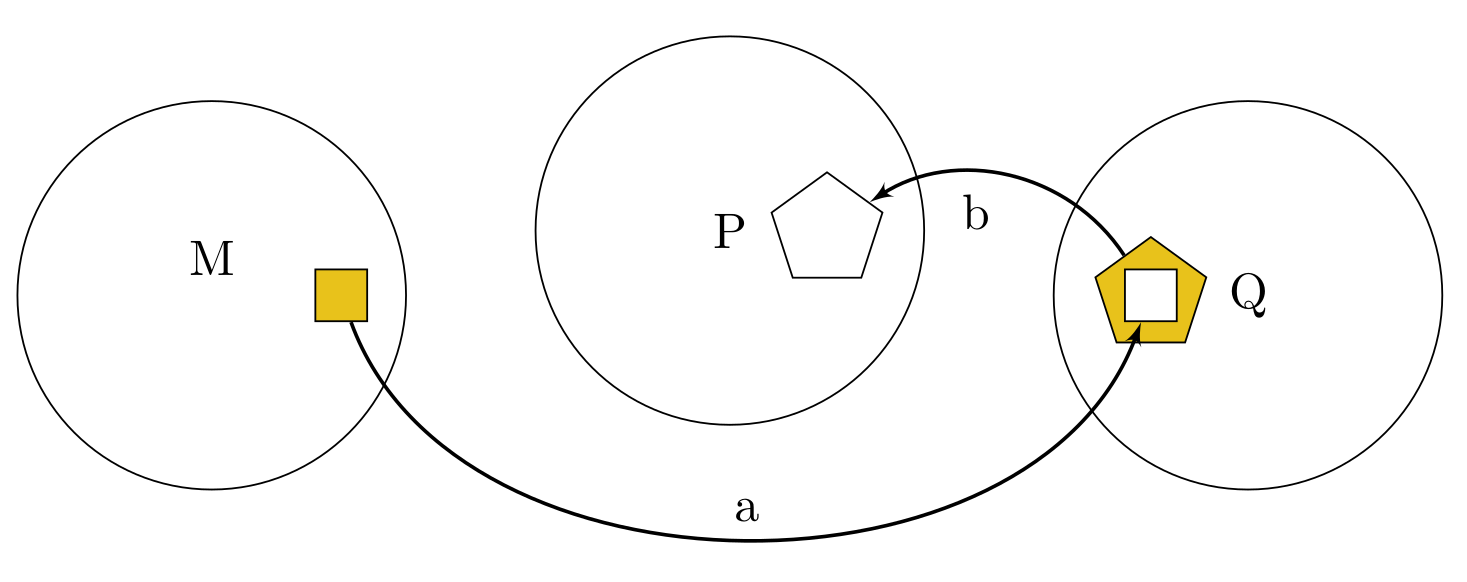} $\qquad\;$
\includegraphics[width=0.4\textwidth]{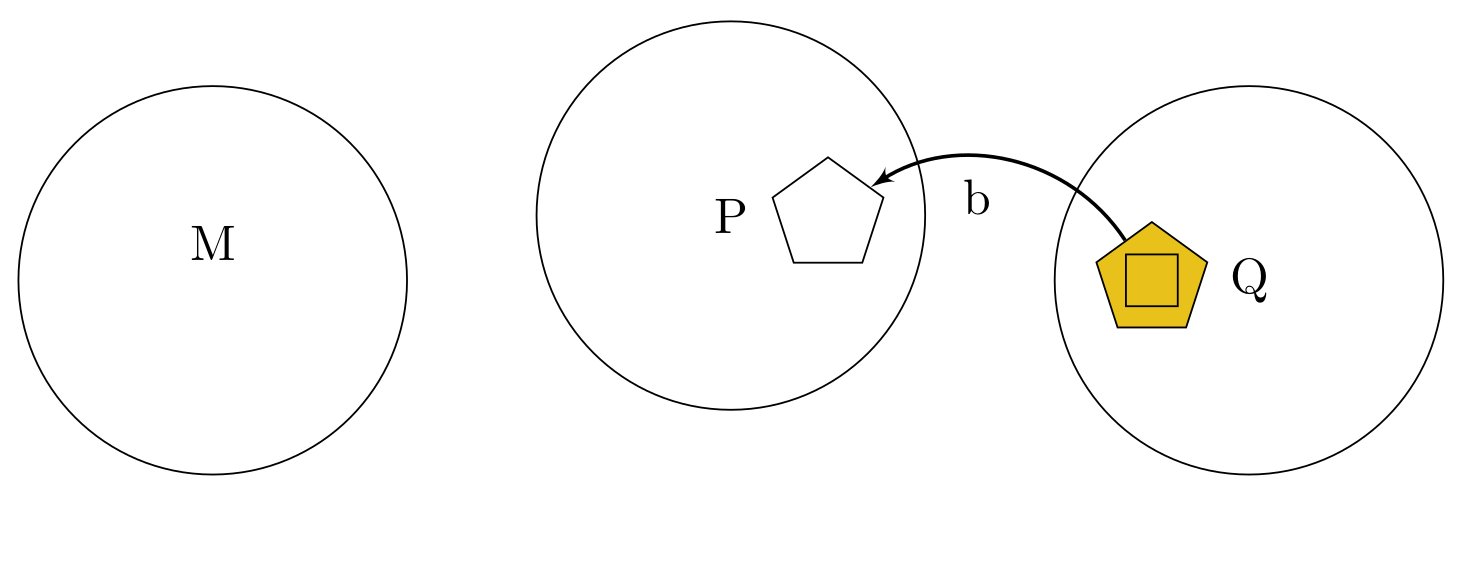}
\vspace{-4pt}
\end{center}The cross reduction handles precisely this kind of missing
arguments. It enables $Q$ to send immediately the message through the
channel $a$ and establishes a new communication channel $c$ on the fly
(below left)
which redirects the missing term, when ready, to the new location of
the message inside $P$ (below right).
\begin{center}
\vspace{-4pt}
  \includegraphics[width=0.4\textwidth]{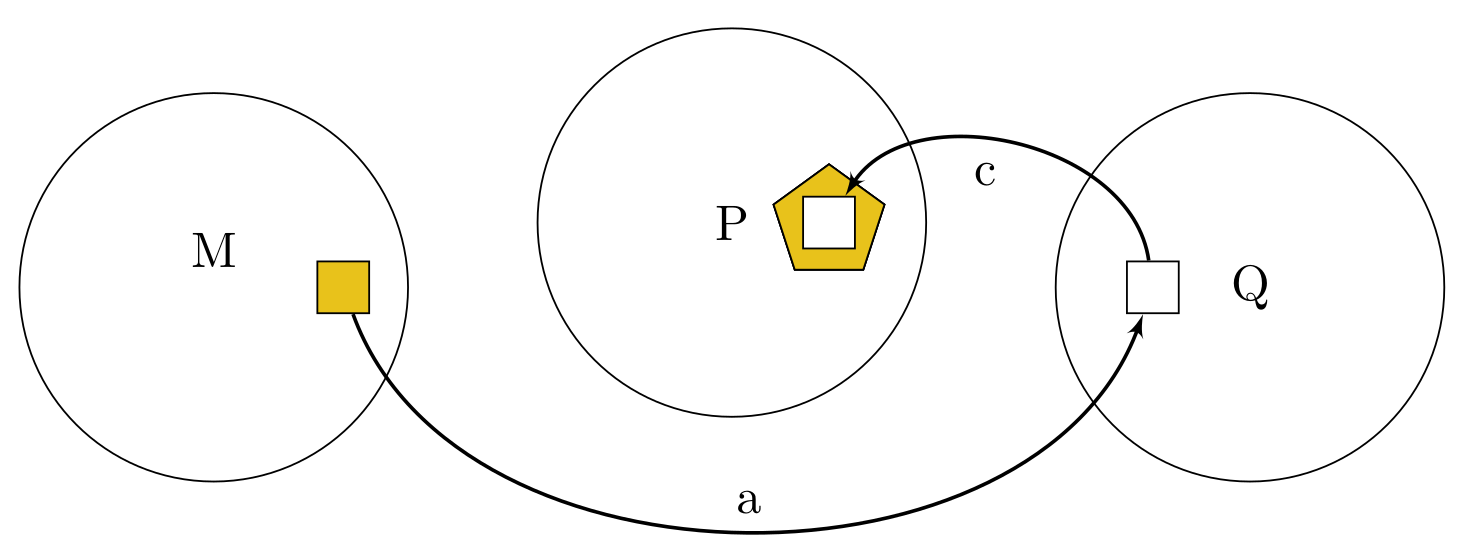} $\qquad\;$
  \includegraphics[width=0.4\textwidth]{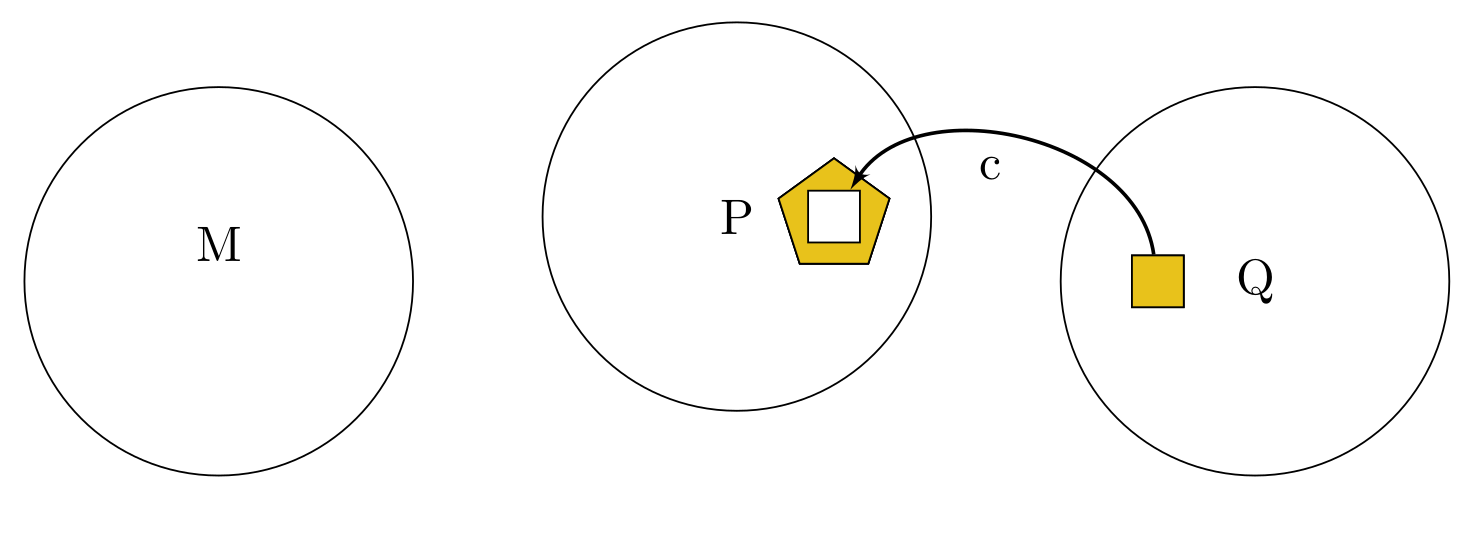}
\vspace{-4pt}
\end{center} We can now partially evaluate $P$,
which in the best case will not even need the yellow square. 

\noindent Both reductions terminate then with $
\quad \vcenter{\includegraphics[width=0.4\textwidth]{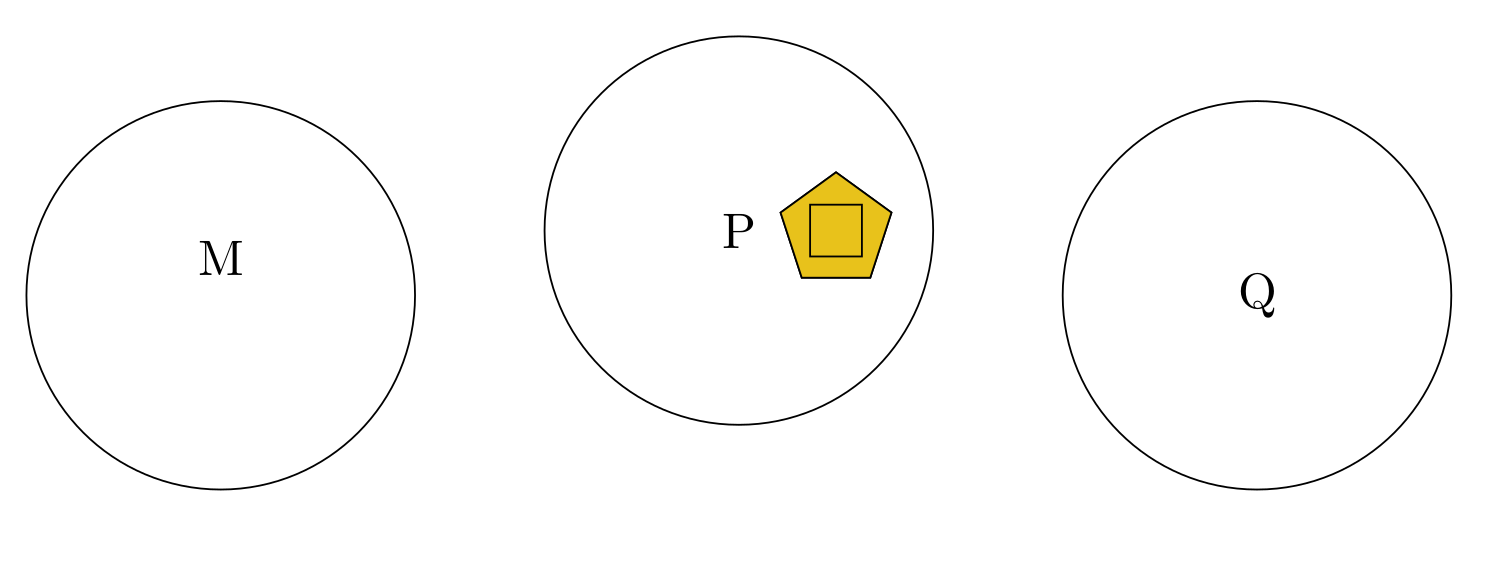}} $ the former,
sending the whole message (yellow pentagon and square) by $b$; the latter, redirecting the missing part of the
message (yellow square) by the new channel  $c$.
 For a concrete example assume that\begin{small}
\[    M \; \mapsto^{*}\; \efq{S}{(\send{a} \, (\lambda x^{T \IMPL \FAL} \, x\, t ))} \qquad\quad
    Q  \; = \;   \efq{S}{(a (\lam y^T \, \send{b} \lan s ,
        y \ran ))}
          \qquad\quad 
    P  \; = \;  b\pi_0 \]
\end{small}where $s:S$ and $t:T$ are closed terms, the complexity of
$S$ is much higher than that of $T$, $b: S
\ET T, \send{b}: \NON (S
\ET T), a: (T \IMPL \FAL)\IMPL \FAL$ and $ \send{a}:\NON ((T \IMPL
\FAL)\IMPL \FAL) $. Without a special mechanism for sending open terms, $Q$ must wait for $M$ to
normalize. Afterwards $M$
sends $\lambda x^{T \IMPL \FAL} \, x\,t$ by $a$ to
$Q$:\begin{small} \begin{align*} & M \parallel_{a} ( Q \parallel_{b}
P) \overset{\mbox{\begin{small}since $a$ is not in
$P$\end{small}}}{\mapsto^{*}} (M \parallel_{a} Q) \parallel_{b} P
\;\mapsto^{*} \; (\efq{S}{(\send{a} \, (\lambda x^{T \IMPL \FAL} \,
x\, t ))} \parallel_{a} Q) \parallel_{b} P \quad \mapsto \\ &
\efq{S}{(( \lambda x^{T \IMPL \FAL} \, x\, t) (\lam y^T \, \send{b}
\lan s , y \ran ))} \p_b P \; \mapsto \; \efq{S}{ ((\lam y^T \,
\send{b} \lan s , y \ran ) t)} \p_b P \, \; \mapsto \;
\efq{S}{(\send{b} \lan s , t \ran)} \p_b b\pi_0 \; \mapsto \; \, \lan
s, t\ran \pi_0 \; \mapsto\; s
  \end{align*} \end{small} Clearly $P$ does not need $t$ at all. Even
though it waited for the pair $\lan s, t\ran $, $P$ only uses the term
$s$.

Our normalization instead enables  $Q$ to directly send
$\lan s, y\ran  $ to $P$ by executing a full cross reduction:
\begin{small}
  \begin{align*} & M \parallel_{a} (Q \parallel_{b} P ) = M \p_a
(\efq{S} {( a (\lam y^T \, \send{b} \lan s , y \ran )) } \p_b b\pi _0)
\quad   \mapsto^{*} \quad M \p_a ((\efq{S} { ( a (\lam y^T \, \send{c} y )) }
\p_b P) \p _c \lan s , c \ran \pi _0 )
  \end{align*}
\end{small}where the channel $c$ handles the redirection of the data $y^T$ in
case it is available later. In our case $P$
already contains all it needs to terminate its computation, indeed
\begin{small}
  \begin{align*}
    \mapsto\quad  M \p_a (( \efq{S} { (   a (\lam y^T \,  \send{c} y 
    ))} \p_b P) \p _c s) \quad \mapsto^* \quad   s
  \end{align*}
\end{small}
since $s$ does not contain communications
anymore. Notice that the time-consuming normalization of the term $M$
does not even
need to be finished at this point.
\end{example}

\subsection*{Conclusions} 

We introduced $\lamem$, a parallel extension of simply typed
$\lam$-calculus.  The calculus $\lamem$ provides a first computational
interpretation of classical proofs as {\em parallel} programs. Our
calculus is defined via Curry--Howard correspondence using a natural
deduction system based on the $\emiddle$ axiom $A \vel \non A$. The
definition of $\lamem$ exploits ideas and techniques developed
in~\cite{lics2017} for the calculus $\lamg$ based on the linearity
axiom, but the specific features of $\emiddle$ made it possible to
define a significantly simpler calculus with more manageable
reductions -- including those for the transmission of closures.  In
spite of its simplicity, the resulting calculus is more expressive
than simply typed $\lam$-calculus and Parigot's $\lam
_\mu$~\cite{Parigot}. Furthermore terms typed by $( \emiddle )$ admit
communication reductions including broadcast communications
and races.

% The shape of $\emiddle$ also enables the definition of communication
% reductions with additional features with respect to $\lamg$, such as
% broadcast communications and races.

Finally, we remark that the permutation reductions of parallel operators
undermine a strong normalization result for the calculus. Indeed, such
reductions enable loops similar to those occurring in cut-elimination
procedures for sequent calculi.  Restrictions on the permutations
might be enough to prove strong normalization, but we leave this as an
open problem.

\bibliographystyle{eptcs}
\bibliography{bibgandalf2018}

\begin{thebibliography}{10}
\providecommand{\bibitemdeclare}[2]{}
\providecommand{\surnamestart}{}
\providecommand{\surnameend}{}
\providecommand{\urlprefix}{Available at }
\providecommand{\url}[1]{\texttt{#1}}
\providecommand{\href}[2]{\texttt{#2}}
\providecommand{\urlalt}[2]{\href{#1}{#2}}
\providecommand{\doi}[1]{doi:\urlalt{http://dx.doi.org/#1}{#1}}
\providecommand{\bibinfo}[2]{#2}

\bibitemdeclare{inproceedings}{lics2017}
\bibitem{lics2017}
\bibinfo{author}{F.~\surnamestart Aschieri\surnameend},
  \bibinfo{author}{A.~\surnamestart Ciabattoni\surnameend} \&
  \bibinfo{author}{F.A. \surnamestart Genco\surnameend} (\bibinfo{year}{2017}):
  \emph{\bibinfo{title}{G{\"{o}}del logic: From natural deduction to parallel
  computation}}.
\newblock In: {\sl \bibinfo{booktitle}{{LICS} 2017}}, pp.
  \bibinfo{pages}{1--12}, \doi{10.1109/LICS.2017.8005076}.

\bibitemdeclare{book}{Barendregt}
\bibitem{Barendregt}
\bibinfo{author}{H.P. \surnamestart Barendregt\surnameend}
  (\bibinfo{year}{1984}): \emph{\bibinfo{title}{The Lambda Calculus, its Syntax
  and Semantics}}.
\newblock \bibinfo{publisher}{Amsterdam: North-Holland},
  \doi{10.1016/c2009-0-14341-6}.

\bibitemdeclare{inproceedings}{Boudol89}
\bibitem{Boudol89}
\bibinfo{author}{G.~\surnamestart Boudol\surnameend} (\bibinfo{year}{1989}):
  \emph{\bibinfo{title}{Towards a lambda-calculus for concurrent and
  communicating systems}}.
\newblock In: {\sl \bibinfo{booktitle}{{TAPSOFT} 1998}}, pp.
  \bibinfo{pages}{149--161}, \doi{10.1007/3-540-50939-9\_130}.

\bibitemdeclare{inproceedings}{Herbelin}
\bibitem{Herbelin}
\bibinfo{author}{P.-L. \surnamestart Curien\surnameend} \&
  \bibinfo{author}{H.~\surnamestart Herbelin\surnameend}
  (\bibinfo{year}{2000}): \emph{\bibinfo{title}{The duality of computation}}.
\newblock In: {\sl \bibinfo{booktitle}{{ICFP} 2000}}, pp.
  \bibinfo{pages}{233--243}, \doi{10.1145/351240.351262}.

\bibitemdeclare{article}{DanosKrivine}
\bibitem{DanosKrivine}
\bibinfo{author}{V.~\surnamestart Danos\surnameend} \& \bibinfo{author}{J.-L.
  \surnamestart Krivine\surnameend} (\bibinfo{year}{2000}):
  \emph{\bibinfo{title}{Disjunctive Tautologies as Synchronisation Schemes}}.
\newblock {\sl \bibinfo{journal}{{CSL 2000}}}, pp. \bibinfo{pages}{292--301},
  \doi{10.1007/3-540-44622-2\_19}.

\bibitemdeclare{inproceedings}{EM}
\bibitem{EM}
\bibinfo{author}{C.~\surnamestart Ene\surnameend} \&
  \bibinfo{author}{T.~\surnamestart Muntean\surnameend} (\bibinfo{year}{1999}):
  \emph{\bibinfo{title}{Expressiveness of point-to-point versus broadcast
  communications}}.
\newblock In: {\sl \bibinfo{booktitle}{FCT 1999}}, pp.
  \bibinfo{pages}{258--268}, \doi{10.1007/3-540-48321-7\_21}.

\bibitemdeclare{inproceedings}{EBPJ2011}
\bibitem{EBPJ2011}
\bibinfo{author}{J.~\surnamestart Epstein\surnameend}, \bibinfo{author}{A.P.
  \surnamestart Black\surnameend} \& \bibinfo{author}{S.L.~Peyton \surnamestart
  Jones\surnameend} (\bibinfo{year}{2011}): \emph{\bibinfo{title}{Towards
  Haskell in the cloud}}.
\newblock In: {\sl \bibinfo{booktitle}{ACM Haskell Symposium 2011}}, pp.
  \bibinfo{pages}{118--129}, \doi{10.1145/2034675.2034690}.

\bibitemdeclare{article}{Fuggetta}
\bibitem{Fuggetta}
\bibinfo{author}{A.~\surnamestart Fuggetta\surnameend}, \bibinfo{author}{G.P.
  \surnamestart Picco\surnameend} \& \bibinfo{author}{G.~\surnamestart
  Vigna\surnameend} (\bibinfo{year}{1998}): \emph{\bibinfo{title}{Understanding
  Code Mobility}}.
\newblock {\sl \bibinfo{journal}{{IEEE} Trans. Software Eng.}}
  \bibinfo{volume}{24}(\bibinfo{number}{5}), pp. \bibinfo{pages}{342--361},
  \doi{10.1109/32.685258}.

\bibitemdeclare{book}{Girard}
\bibitem{Girard}
\bibinfo{author}{J.-Y. \surnamestart Girard\surnameend},
  \bibinfo{author}{Y.~\surnamestart Lafont\surnameend} \&
  \bibinfo{author}{P.~\surnamestart Taylor\surnameend} (\bibinfo{year}{1989}):
  \emph{\bibinfo{title}{Proofs and Types}}.
\newblock \bibinfo{publisher}{Cambridge University Press}.
\newblock \urlprefix\url{http://www.paultaylor.eu/stable/prot.pdf}.

\bibitemdeclare{inproceedings}{Griffin}
\bibitem{Griffin}
\bibinfo{author}{T.G. \surnamestart Griffin\surnameend} (\bibinfo{year}{1990}):
  \emph{\bibinfo{title}{A Formulae-as-Type Notion of Control}}.
\newblock In: {\sl \bibinfo{booktitle}{{POPL 1990}}}, pp.
  \bibinfo{pages}{47--58}, \doi{10.1145/96709.96714}.

\bibitemdeclare{inproceedings}{deGrooteex}
\bibitem{deGrooteex}
\bibinfo{author}{P.~\surnamestart de~Groote\surnameend} (\bibinfo{year}{1995}):
  \emph{\bibinfo{title}{A Simple Calculus of Exception Handling}}.
\newblock In: {\sl \bibinfo{booktitle}{{TLCA 1995}}}, pp.
  \bibinfo{pages}{201--215}, \doi{10.1007/BFb0014054}.

\bibitemdeclare{inproceedings}{HT}
\bibitem{HT}
\bibinfo{author}{K.~\surnamestart Honda\surnameend} \&
  \bibinfo{author}{M.~\surnamestart Tokoro\surnameend} (\bibinfo{year}{1991}):
  \emph{\bibinfo{title}{An Object Calculus for Asynchronous Communication}}.
\newblock In: {\sl \bibinfo{booktitle}{{ECOOP} 1991}}, pp.
  \bibinfo{pages}{133--147}, \doi{10.1007/BFb0057019}.

\bibitemdeclare{inproceedings}{Howard}
\bibitem{Howard}
\bibinfo{author}{W.A. \surnamestart Howard\surnameend} (\bibinfo{year}{1980}):
  \emph{\bibinfo{title}{The formulae-as-types notion of construction}}.
\newblock In: {\sl \bibinfo{booktitle}{To H. B. Curry: Essays on Combinatory
  Logic, Lambda Calculus, and Formalism}}, \bibinfo{publisher}{Academic Press},
  pp. \bibinfo{pages}{479--491}.

\bibitemdeclare{incollection}{Krivine1}
\bibitem{Krivine1}
\bibinfo{author}{J.-L. \surnamestart Krivine\surnameend}
  (\bibinfo{year}{1990}): \emph{\bibinfo{title}{Lambda-calcul types et
  mod\`eles}}.
\newblock In: {\sl \bibinfo{booktitle}{Studies in Logic and Foundations of
  Mathematics}}, \bibinfo{publisher}{Masson}, pp. \bibinfo{pages}{1--176}.

\bibitemdeclare{article}{Krivine}
\bibitem{Krivine}
\bibinfo{author}{J.-L. \surnamestart Krivine\surnameend}
  (\bibinfo{year}{2009}): \emph{\bibinfo{title}{Realizability in classical
  logic}}.
\newblock {\sl \bibinfo{journal}{Panoramas et synth\`eses}}, pp.
  \bibinfo{pages}{197--229}.
\newblock \urlprefix\url{https://hal.archives-ouvertes.fr/hal-00154500}.

\bibitemdeclare{article}{Milner}
\bibitem{Milner}
\bibinfo{author}{R.~\surnamestart Milner\surnameend} (\bibinfo{year}{1992}):
  \emph{\bibinfo{title}{Functions as Processes}}.
\newblock {\sl \bibinfo{journal}{Mathematical Structures in Computer Science}}
  \bibinfo{volume}{2}(\bibinfo{number}{2}), pp. \bibinfo{pages}{119--141},
  \doi{10.1017/S0960129500001407}.

\bibitemdeclare{article}{Parigot}
\bibitem{Parigot}
\bibinfo{author}{M.~\surnamestart Parigot\surnameend} (\bibinfo{year}{1997}):
  \emph{\bibinfo{title}{Proofs of Strong Normalization for Second-Order
  Classical Natural Deduction}}.
\newblock {\sl \bibinfo{journal}{J.\ Symbolic Logic}}
  \bibinfo{volume}{62}(\bibinfo{number}{4}), pp. \bibinfo{pages}{1461--1479},
  \doi{10.2307/2275652}.

\bibitemdeclare{inproceedings}{Prawitz}
\bibitem{Prawitz}
\bibinfo{author}{D.~\surnamestart Prawitz\surnameend} (\bibinfo{year}{1971}):
  \emph{\bibinfo{title}{Ideas and Results in Proof Theory}}.
\newblock In: {\sl \bibinfo{booktitle}{Proceedings of the Second Scandinavian
  Logic Symposium}}, pp. \bibinfo{pages}{235--307}, \doi{10.2307/2271904}.

\bibitemdeclare{book}{sangiorgiwalker2003}
\bibitem{sangiorgiwalker2003}
\bibinfo{author}{D.~\surnamestart Sangiorgi\surnameend} \&
  \bibinfo{author}{D.~\surnamestart Walker\surnameend} (\bibinfo{year}{2003}):
  \emph{\bibinfo{title}{The pi-calculus: a Theory of Mobile Processes}}.
\newblock \bibinfo{publisher}{Cambridge University Press}.

\bibitemdeclare{book}{Sorensen}
\bibitem{Sorensen}
\bibinfo{author}{M.H.B. \surnamestart S{\o}rensen\surnameend} \&
  \bibinfo{author}{P.~\surnamestart Urzyczyn\surnameend}
  (\bibinfo{year}{1998}): \emph{\bibinfo{title}{Lectures on the Curry-Howard
  Isomorphism}}.
\newblock \bibinfo{publisher}{Elsevier}, \doi{10.1016/s0049-237x(06)80005-4}.

\bibitemdeclare{inproceedings}{TCP2013}
\bibitem{TCP2013}
\bibinfo{author}{B.~\surnamestart Toninho\surnameend},
  \bibinfo{author}{L.~\surnamestart Caires\surnameend} \&
  \bibinfo{author}{F.~\surnamestart Pfenning\surnameend}
  (\bibinfo{year}{2013}): \emph{\bibinfo{title}{Higher-Order processes,
  functions, and sessions: a monadic integration}}.
\newblock In: {\sl \bibinfo{booktitle}{{ESOP} 2013}}, pp.
  \bibinfo{pages}{350--369}, \doi{10.1007/978-3-642-37036-6_20}.

\bibitemdeclare{article}{WadlerSQ}
\bibitem{WadlerSQ}
\bibinfo{author}{P.~\surnamestart Wadler\surnameend} (\bibinfo{year}{2003}):
  \emph{\bibinfo{title}{Call-by-value is dual to call-by-name}}.
\newblock {\sl \bibinfo{journal}{{SIGPLAN} Notices}}
  \bibinfo{volume}{38}(\bibinfo{number}{9}), pp. \bibinfo{pages}{189--201},
  \doi{10.1145/944746.944723}.

\bibitemdeclare{article}{Wadler2012}
\bibitem{Wadler2012}
\bibinfo{author}{P.~\surnamestart Wadler\surnameend} (\bibinfo{year}{2012}):
  \emph{\bibinfo{title}{Propositions as Sessions}}.
\newblock {\sl \bibinfo{journal}{J. of Functional Programming}}
  \bibinfo{volume}{24}, pp. \bibinfo{pages}{384--418},
  \doi{10.1145/2398856.2364568}.

\bibitemdeclare{article}{Wadler}
\bibitem{Wadler}
\bibinfo{author}{P.~\surnamestart Wadler\surnameend} (\bibinfo{year}{2015}):
  \emph{\bibinfo{title}{Propositions as Types}}.
\newblock {\sl \bibinfo{journal}{Communications of the ACM}}
  \bibinfo{volume}{58}(\bibinfo{number}{12}), pp. \bibinfo{pages}{75--84},
  \doi{10.1145/2699407}.

\end{thebibliography}

\noappendix{
%%%% to comment appendix

\newpage
\appendix
\section{Appendix}

%%%%INIZIO  1

\noindent \textbf{Propostition~\ref{propositionstrongsubf}}~(Characterization of Strong Subformulas)\textbf{.}
% \begin{proposition}[Characterization of Strong Subformulas]
%   \ref{proposition-strongsubf}
Suppose $B$ is any {strong subformula} of $A$. Then:
\begin{itemize}
\item If $A=A_{1}\land \ldots \land A_{n}$, with $n>0$ and $A_{1}, \ldots, A_{n}$ are prime, then $B$ is a proper subformula of one among $A_{1}, \ldots, A_{n}$. 
\item If $A=C\rightarrow D$, then $B$ is a proper subformula of a prime factor of $C$ or $D$. 
\end{itemize}
% \end{proposition}
\begin{proof}\mbox{}
\begin{itemize}
\item Suppose $A=A_{1}\land \ldots \land A_{n}$, with $n>0$ and $A_{1}, \ldots, A_{n}$ are prime. Any prime proper subformula of $A$ is a subformula of one among $A_{1}, \ldots, A_{n}$, so $B$ must be a proper subformula of one among $A_{1}, \ldots, A_{n}$.

\item Suppose  $A=C\rightarrow D$. Any prime proper subformula $\mathcal{X}$ of $A$ is first of all a subformula of $C$ or $D$. Assume now $C=C_{1}\land \ldots \land C_{n}$ and $D=D_{1}\land \ldots \land D_{m}$, with $C_{1}, \ldots, C_{n}, D_{1}, \ldots, D_{m}$ prime. Since  $\mathcal{X}$ is prime, it must be a subformula of one among $C_{1}, \ldots, C_{n}, D_{1}, \ldots, D_{m}$ and since $B$ is a proper subformula of $\mathcal{X}$, it must be a proper subformula of one among $C_{1}, \ldots, C_{n}, D_{1}, \ldots, D_{m}$. 
\end{itemize}
\end{proof}

\noindent \textbf{Proposition~\ref{prop:boundhyp}}~(Bound Hypothesis Property)\textbf{.}
% \begin{proposition}[Bound Hypothesis Property]
%   \ref{proposition-boundhyp}
Suppose
\[x_{1}^{A_{1}}, \ldots, x_{n}^{A_{n}}\vdash t: A\]
$t\in\nf$ is a simply typed $\lambda$-term and  $z: B$ a variable occurring bound in $t$. Then one of the following holds:
\begin{enumerate}
\item $B$ is a proper subformula of a prime factor of $A$.
\item $B$ is a strong subformula of one among $A_{1},\ldots, A_{n}$.
%There is a number $k$, with $1\leq k\leq n$, such that $B$ is a proper subformula of $A_{k}$; moreover, %{\color{red}{if $A_{k}$ is a conjunction, then $B$ is a proper subformula of a prime factor of $A_{k}$, and}} if $A_{k}=C\rightarrow D$, then $B$ is a proper subformula of a prime factor of $C$ or $D$. 
\end{enumerate}
% \end{proposition}
\begin{proof}
By induction on $t$.  
\begin{itemize} 
\item $t=x_{i}^{A_{i}}$, with $1\leq i\leq n$. Since by hypothesis $z$ must occur bound in $t$,  this case is impossible.

\item $t=\lambda x^{T} u$, with $A=T\rightarrow U$. If $z=x^{T}$,
since $A$ is a prime factor of itself, we are done. If $z\neq x^{T}$,
then $z$ occurs bound in $u$ and by induction hypothesis applied to
$u: U$, we have two possibilities: i) $B$ is a proper subformula of
a prime factor of $U$ and thus a proper subformula of a prime factor
-- $A$ itself -- of $A$; ii) $B$ already satisfies 2., and we are
done, or $B$ is a strong subformula of $T$, and thus it satisfies 1.

\item $t=\langle u_{1}, u_{2}\rangle$, with $A=T_{1}\land T_{2}$. Then
$z$ occurs bound in $u$ or $v$ and, by induction hypothesis applied to
$u_{1}: T_{1}$ and $u_{2}: T_{2}$, we have two possibilities: i) $B$
is a proper subformula of a prime factor of $T_{1}$ or $T_{2}$, and
thus $B$ is a proper subformula of a prime factor of $A$ as well; ii)
$B$ satisfies 2. and we are done.

\item $t= \efq{P}{u}$, with $A=P$. Then $z$ occurs bound in $u$. Since
$\bot$ has no proper subformula, by induction hypothesis applied to
$u: \bot$, we have that $B$ satisfies 2. and we are done.

\item $t=x_{i}^{A_{i}}\, \xi_{1}\ldots \xi_{m}$, where $m>0$ and each
$\xi_{j}$ is either a term or a projection $\pi_{k}$. Since $z$ occurs
bound in $t$, it occurs bound in some term $\xi_{j}: T$, where $T$ is
a proper subformula of $A_{i}$. By induction hypothesis applied to
$\xi_{j}$, we have two possibilities: i) $B$ is a proper subformula of
a prime factor of $T$ and, by def.~\ref{definitionstrongsubf},
$B$ is a strong subformula of $A_{i}$.
%Suppose moveover $A_{i}=C\rightarrow D$. Then either $i=1$, so $T=C$, and 1. holds, or $i>1$, thus $T$ is a subformula of a prime factor of $D$, and 1. holds again.
ii) $B$ satisfies 2. and we are done.
\end{itemize}
\end{proof}

 \noindent \textbf{Proposition~\ref{prop:app}}\textbf{.}
 % \begin{proposition}
 % \ref{prop:app}
   Suppose that $t\in \nf$ is a simply typed $\lambda$-term and
 \[x_{1}^{A_{1}}, \ldots, x_{n}^{A_{n}}, z^{B}\vdash t: A\]
 Then one of the following holds:
 % is true:
 \begin{enumerate}
 \item \emph{Every occurrence of $z^{B}$ in $t$ is of the form
     $z^{B}\, \xi$ for some proof term or projection $\xi$.}
 \item \emph{$B=\bot$ or $B$ is a subformula of $A$ or a proper
     subformula of one among the formulas $A_{1}, \ldots, A_{n}$.}
 \end{enumerate}
 % \end{proposition}

 \begin{proof}
 By induction on $t$. 
 %We reason by cases on
 % , according to 
 %the shape of $t$. 
 \begin{itemize} 
 \item $t=x_{i}^{A_{i}}$, with $1\leq i\leq n$. Trivial.
 \item $t=z^{B}$. This means that $B=A$, and we are done.
 \item $t=\lambda x^{T} u$, with $A=T\rightarrow U$. By induction hypothesis applied to $u: U$, we have two possibilities:  i)  every occurrence of $z^{B}$ in $u$ is of the form $z^{B}\, \xi$, and we are done; ii) $B=\bot$ or $B$ is a  subformula of  $U$,  and hence of $A$, or a  proper subformula of one among the formulas $A_{1}, \ldots, A_{n}$, and we are done again.
 \item $t=\langle u_{1}, u_{2}\rangle$, with $A=T_{1}\land T_{2}$. By
 induction hypothesis applied to $u_{1}: T_{1}$ and $u_{2}: T_{2}$, we
 have two possibilities: i) every occurrence of $z^{B}$ in $u_{1}$ and
 $u_{2}$ is of the form $z^{B}\, \xi$, and we are done; ii) $B=\bot$ or
 $B$ is a subformula of $T_{1}$ or $T_{2}$, and hence of $A$, or a
 proper subformula of one among the formulas $A_{1}, \ldots, A_{n}$,
 and we are done again.
 \item $t= \efq{P}{u}$, with $A=P$.  By induction hypothesis applied to $u: \bot$, we have two possibilities: i)  every occurrence of $z^{B}$ in $u$ is of the form $z^{B}\, \xi$, and we are done; ii) $B=\bot$  or a  proper subformula of one among 
 % the formulas
   $A_{1}, \ldots, A_{n}$, and we are done again.
 \item $t=x_{i}^{A_{i}}\, \xi_{1}\ldots \xi_{m}$, where $m>0$ and each
 $\xi_{j}$ is either a term or a projection $\pi_{k}$. Suppose there is
 an $i$ such that in the term $\xi_{j}: T_{j}$ not every occurrence of
 $z^{B}$ in $u$ is of the form $z^{B}\, \xi$. If $B=\bot$, we are
 done. If not, then by induction hypothesis $B$ is a subformula of
 $T_{j}$ or a proper subformula of one among $A_{1}, \ldots,
 A_{n}$. Since $T_{j}$ is a proper subformula of $A_{i}$, in both cases
 $B$ is a proper subformula of one among $A_{1}, \ldots, A_{n}$.
 \item $t=z^{B}\, \xi_{1}\ldots \xi_{m}$, where $m>0$ and each
 $\xi_{i}$ is either a term or a projection $\pi_{j}$. Suppose there is
 an $i$ such that in the term $\xi_{i}: T_{i}$ not every occurrence of
 $z^{B}$ in $u$ is of the form $z^{B}\, \xi$. If $B=\bot$, we are
 done. If not, then by induction hypothesis $B$ is a subformula of
 $T_{i}$ or a proper subformula of one among $A_{1}, \ldots,
 A_{n}$. But the former case is not possible, since $T_{i}$ is a proper
 subformula of $B$, hence the latter holds.
 \end{itemize}
 \end{proof}

% }

% %%%% FINE  4

%%%% INIZIO  5

\noindent \textbf{Proposition~\ref{prop:parallelform}}~(Parallel Form Property)\textbf{.}

% \begin{proposition}[Parallel Normal Form Property]
%   \ref{prop:parallelform}
  Suppose $t\in \nf$, then it is parallel form.
  \begin{enumerate}
  \item \emph{Every occurrence of $z^{B}$ in $t$ is of the form
      $z^{B}\, \xi$ for some proof term or projection $\xi$.}
  \item \emph{$B=\bot$ or $B$ is a subformula of $A$ or a proper
      subformula of one among the formulas $A_{1}, \ldots, A_{n}$.}
  \end{enumerate}
% \end{proposition}
\begin{proof}
By induction on $t$. %We reason by cases on
% , according to
%the shape of $t$. 
\begin{itemize}

\item $t$ is a variable $x$. Trivial. 

\item $t=\lambda x\, v$. Since $t$ is normal, $v$ cannot be of the
form $u_1 \p_a u_2$, otherwise one could apply the
permutation \[t= \lam x^{A} \, u_1 \p_a u_2 \mapsto
\lam x^{A} \, u_{1} \p_a \lam x^{A} \, u_2 \] and $t$
would not be in normal form. Hence, by induction hypothesis $v$ must
be a simply typed $\lambda$-term.

\item $t=\langle v_{1}, v_{2}\rangle$. Since $t$ is normal, neither
$v_{1}$ nor $v_{2}$ can be of the form $u_1 \p_a u_2$,
otherwise one could apply one of the permutations
\[\langle u_1 \p_a u_2 , \, w\rangle \mapsto
\lan u_1 , w \ran  \p_a \lan u_2, w\ran \]
\[\langle w, \,u_1 \p_a u_2\rangle \mapsto
\lan w, u_{1}\ran \p_a \lan w, u_{2}\ran\] and $t$ would
not be in normal form. Hence, by induction hypothesis $v_{1}$ and
$v_{2}$ must be simply typed $\lambda$-terms.

\item $t=v_1 \, v_2$. Since $t$ is normal, neither $v_1$ nor $v_2$ can
be of the form $u_1\p_a u_2$, otherwise one could
apply one of the permutations \[w\, (u_1\p_a u_2 )
\mapsto wu_1\p_a wu_2  \]  \[(u_1\p_a u_2 ) \, w
\mapsto u_1w\p_a u_2  w \]
and $t$ would
not be in normal form. Hence, by induction hypothesis $v_{1}$ and
$v_{2}$ must be simply typed $\lambda$-terms.

\item $t=  \efq{P}{v}$. Since $t$ is normal, $v$ cannot
be of the form $u_1\p_a u_2$, otherwise one could apply the permutation
\[\efq{P}{u_1\p_a u_2} \mapsto
\efq{P}{u_{1}}\p_a \efq{P}{u_{2}}\] and $t$ would not be in
normal form. Hence, by induction hypothesis $u_{1}$ and $u_{2}$ must
be simply typed $\lambda$-terms.

\item $t=u\, \pi_{i}$. Since $t$ is normal, $v$ can
be of the form $u_1 \p_a u_2$, otherwise one could apply the permutation
\[(u_1 \p_a u_2 )\, \pi_{i} \mapsto u_{1}\pi_{i}\p_a
 u_2\pi_{i}\] and $t$ would not be in normal form. Hence, by
induction hypothesis $u$ must be a simply typed $\lambda$-term, which
is the thesis.

\item $t= u_1 \p_a u_2$. By induction hypothesis the thesis holds
for $u_i$ where $1 \leq i \leq 2$ and hence trivially for $t$.
\end{itemize}
\end{proof}

%%%% FINE 5

%%%%% INIZIO 7

\noindent \textbf{Proposition~\ref{propositionnormpar}}\textbf{.}
% \begin{proposition}\ref{propositionnormpar}
Let $t: A$ be any term. Then $t\mapsto^{*} t'$, where $t'$ is a parallel form. 
% \end{proposition}
\begin{proof}
By induction on $t$.  As a shortcut, if a term $u$ reduces to a term
$u'$ that can be written as $u''$ omitting parentheses and the
subscript of $\p$ operators, we write $u \mapstopar^{*} u''$.
\begin{itemize}
\item  $t$ is a variable $x$. Trivial. 
\item $t=\lambda x\, u$. By induction hypothesis, 
\[u\mapstopar^{*} u_{1}\parallel u_{2}\parallel \ldots \parallel u_{n+1}\]
and each term $u_{i}$, for $1\leq i\leq n+1$, is a simply typed $\lambda$-term. Applying $n$ times the permutations
we obtain
\[t\mapstopar^{*} \lambda x\, u_{1}\parallel \lambda x\, u_{2}\parallel
  \ldots \parallel \lambda x\, u_{n+1}\]
which is the thesis.
\item $t=u\, v$. By induction hypothesis, 
\[u\mapstopar^{*} u_{1}\parallel u_{2}\parallel \ldots \parallel u_{n+1}\]
\[v\mapstopar^{*} v_{1}\parallel v_{2}\parallel \ldots \parallel v_{m+1}\]
and each term $u_{i}$ and $v_{i}$, for $1\leq i\leq n+1, m+1$, is a
simply typed $\lambda$-term. Applying $n+m$ times the permutations we obtain
\[
\begin{aligned}
t &\mapstopar^{*} (u_{1}\parallel u_{2}\parallel \ldots \parallel  u_{n+1})\, v \\
&\mapstopar^{*}  u_{1}\, v \parallel  u_{2}\, v \parallel
\ldots \parallel  u_{n+1}\, v\\
&\mapstopar^{*} u_{1}\, v_{1} \parallel u_{1}\, v_{2}\parallel
\ldots \parallel u_{1}\, v_{m+1} \parallel \ldots
\\
& \qquad  \, \ldots \parallel u_{n+1}\, v_{1} \parallel  u_{n+1}\,
v_{2} \parallel \ldots
 \parallel  u_{n+1}\, v_{m+1}
\end{aligned}
\]

\item $t=\langle u, v\rangle$. By induction hypothesis, 
\[u\mapstopar^{*} u_{1}\parallel  u_{2}\parallel \ldots \parallel u_{n+1}\]
\[v\mapstopar^{*} v_{1}\parallel v_{2}\parallel \ldots \parallel v_{m+1}\]
and each term $u_{i}$ and $v_{i}$, for $1\leq i\leq n+1, m+1$, is a
simply typed $\lambda$-term. Applying $n+m$ times the permutations we
obtain
  \[
    \begin{aligned}
      t &\mapstopar^{*}\langle u_{1}\parallel u_{2}\parallel
      \ldots \parallel  u_{n+1},\, v \rangle\\
      &\mapstopar^{*} \langle u_{1}, v\rangle \parallel \langle u_{2}, v
      \rangle\parallel \ldots \parallel \langle u_{n+1}, v\rangle\\
      &\mapstopar^{*} \langle u_{1}, v_{1}\rangle \parallel
      \langle u_{1}, v_{2}
      \rangle\parallel \ldots \parallel \langle u_{1},
      v_{m+1}\rangle \parallel  \ldots
      \\
      & \qquad \, \ldots
      \parallel \langle u_{n+1},
      v_{1}\rangle \parallel \langle u_{n+1}, v_{2}
      \rangle\parallel \ldots
      \\
      & \qquad \, \ldots \parallel \langle u_{n+1},
      v_{m+1}\rangle
    \end{aligned}
  \]

\item $t=u\, \pi_{i}$. By induction hypothesis,
\[u\mapstopar^{*} u_{1}\parallel  u_{2}\parallel \ldots \parallel u_{n+1}\]
and each term $u_{i}$, for $1\leq i\leq n+1$, is a simply typed
$\lambda$-term. Applying $n$ times the permutations we obtain
\[t\mapstopar^{*}  u_{1}\, \pi_{i}\parallel  u_{2}\,
\pi_{i}\parallel \ldots \parallel u_{n+1} \, \pi_{i}.\]

\item $t= \efq{P}{u} $. By induction hypothesis,
\[u\mapstopar^{*} u_{1}\parallel  u_{2}\parallel \ldots \parallel u_{n+1}\]
and each term $u_{i}$, for $1\leq i\leq n+1$, is a simply typed
$\lambda$-term. Applying $n$ times the permutations we obtain
\[t\mapstopar^{*} \efq{P}{u_{1}} \parallel \efq{P}{u _2 } \parallel
\ldots \parallel \efq{P}{u_{n+1} }\]
\end{itemize}
\end{proof}

\noindent \textbf{Theorem~\ref{thm:subf}}~(Subformula Property)\textbf{.}
% \begin{theorem}[Subformula Property]\ref{thm:subf}
Suppose
\[x_{1}^{A_{1}}, \ldots, x_{n}^{A_{n}}\vdash t: A \quad \mbox{and} \quad
t\in \nf. \quad \mbox{Then}:\]
\begin{enumerate}
\item For each communication variable $a$ occurring bound in  $t$ and with
    communication kind $C$, the prime factors of
$C$  are proper subformulas of  
$A_{1}, \ldots, A_{n}, A$. 
\item The type of any subterm of $t$ which is not a bound communication variable is either a subformula or a conjunction of subformulas of the formulas $A_{1}, \ldots, A_{n}, A$. 
\end{enumerate}
\begin{proof} 
We proceed by induction on $t$. 
By Prop.~\ref{prop:parallelform} if we remove the
subscripts $t = t_{1}\parallel  t_{2}\parallel \ldots \parallel t_{n+1}$
and each $t_{i}$, for $1\leq i\leq n+1$, is a simply typed
$\lambda$-term.

By induction on $t$. 

\begin{itemize} 
\item $t=x_{i}^{A_{i}}$, with $1\leq i\leq n$. Trivial. 

%\item $t=\lambda x^{T} u$, with $A=T\rightarrow U$. By induction hypothesis applied to $u: U$, for each communication variable $a^{B\rightarrow C}$ of $u$, the prime factors of $B$ and $C$ are  proper subformulas of the formulas $T, A_{1},  \ldots, A_{n},  U$ and thus of the formulas $A_{1},  \ldots, A_{n},  A$; moreover, the type of any subterm of $u$ which is not a communication variable is either a subformula or a conjunction of subformulas of the formulas $T, A_{1}, \ldots, A_{n},  U$ and hence of the formulas $A_{1},  \ldots, A_{n},  A$.\\
\item $t=\lambda x^{T} u$, with $A=T\rightarrow U$. By Prop.~\ref{prop:parallelform}, $t$ is a simply typed $\lambda$-term, so $t$ contains no bound communication variable. Moreover, by induction hypothesis applied to $u: U$, the type of any subterm of $u$ which is not a bound communication variable is either a subformula or a conjunction of subformulas of the formulas $T, A_{1}, \ldots, A_{n},  U$ and hence of the formulas $A_{1}, \ldots, A_{n},  A$.

%\item $t=\langle u_{1}, u_{2}\rangle$, with $A=T_{1}\land T_{2}$. By induction hypothesis applied to $u_{1}: T_{1}$ and $u_{2}: T_{2}$, for each communication variable $a^{B\rightarrow C}$ of $t$, the prime factors of $B$ and $C$ are proper  subformulas of the formulas $ A_{1},  \ldots, A_{n}, T_{1}, T_{2}$ and hence of the formulas $A_{1},  \ldots, A_{n},  A$; moreover, the type of any subterm of $u$ which is not a communication variable is either a subformula or a conjunction of subformulas of the formulas $ A_{1}, \ldots, A_{n}, T_{1}, T_{2}$ and hence of the formulas $A_{1},  \ldots, A_{n},  A$.\\

\item $t=\langle u_{1}, u_{2}\rangle$, with $A=T_{1}\land T_{2}$. By Prop.~\ref{prop:parallelform}, $t$ is a simply typed $\lambda$-term, so $t$ contains no bound communication variable. Moreover, by induction hypothesis applied to $u_{1}: T_{1}$ and $u_{2}: T_{2}$, the type of any subterm of $u$ which is not a bound communication variable is either a subformula or a conjunction of subformulas of the formulas $ A_{1}, \ldots, A_{n}, T_{1}, T_{2}$ and hence of the formulas $A_{1}, \ldots, A_{n},  A$.

\item $t =  u_1\p_a u_2  $. Let      $C$ be the communication
kind of $b$, we first show that the communication complexity of $b$ is $0$.
  We reason by contradiction and assume that it is greater than $0$. 
 $u_1,u_2$ are either  simple parallel terms or of the form $v_1 \p_cv_2
$. The second case is not possible, otherwise a permutation reduction
could be applied to $t\in \nf$. Thus $u_{1},u_{2}$ are   simple parallel
terms. Since the communication complexity of $b$ is
greater than $0$, the types of the occurrences of $b$ in $u_{1}, u_{2}$ are not subformulas of $A_{1}, \ldots, A_{n}, A$. 
 By
 Prop.~\ref{prop:app}, all occurrences of $b$ in $u_{1}
, u_{2}$ are of the form $\send{b} w $ for some term $w$.
Hence,  we can write
\[ t= (\dots \p {\mathcal C} [\send{b}\,t]\p \dots) \p_b {\mathcal
    D} \] where ${\mathcal C} [\; ]$ is a simple context, $(\dots \p
{\mathcal C} [\send{b}\,t]\p \dots)$ is a normal simple parallel term,  ${\mathcal
  D}$ is a normal term, and $b$ is
 rightmost in $  {\mathcal C} [\send{b}\,t] $.  Hence a cross reduction of $t$ can be performed, which contradicts
the fact that $t\in\nf$.  Since we have established that the
communication complexity of $b$ is $0$,      the prime factors of $C$  must be proper subformulas of $A_{1}, \ldots, A_{n},
A$. Now, by induction hypothesis applied to $u_{1}:A,u_{2}:A$, for each communication variable  $a^{F}$
 occurring bound in $t$,      the prime factors of $F$ are proper
subformulas of the formulas $ A_{1}, \ldots, A_{n}, A $ or    subformulas
of $ C$   and thus of the formulas $A_{1}, \ldots, A_{n}, A$;
moreover, the type of any subterm of $u_{1}, u_{2}$ which is
not a communication variable is either a subformula or a  conjunction
of subformulas of the formulas   $ A_{1}, \ldots, A_{n}, C$  and thus of $A_{1}, \ldots, A_{n}, A$.

\item $t=x_{i}^{A_{i}}\, \xi_{1}\ldots \xi_{m}$, where $m>0$ and each
$\xi_{j}: T_{j}$ is either a term or a projection $\pi_{k}$ and
$T_{j}$ is a subformula of $A_{i}$. By Prop.~\ref{prop:parallelform}, $t$ is a simply typed $\lambda$-term,
so $t$ contains no bound communication variable. By induction
hypothesis applied to each $\xi_{j}: T_{j}$, the type of any subterm
of $t$ which is not a bound communication variable is either a
subformula or a conjunction of subformulas of the formulas $ A_{1},
\ldots, A_{n}, T_{1}, \ldots, T_{m}$ and thus of the formulas
$A_{1},  \ldots, A_{n}, A$.

\item $t= \efq{P}{u}$, with $A=P$. Then $u=x_{i}^{A_{i}}\, \xi_{1}\ldots \xi_{m}$, where $m>0$ and each $\xi_{j}$ is either a term or a projection $\pi_{k}$. Hence, $\bot$ is a subformula of $A_{i}$. Finally, by the proof of the previous case, we obtain the thesis for $t$.
\end{itemize}
\end{proof}

%%%%% FINE 6

}
%%% to comment appendix

\end{document}